\documentclass[journal,draftcls, onecolumn]{IEEEtran}
\usepackage[dvipdfmx]{graphicx}
\usepackage{amsmath}
\usepackage{amssymb}
\usepackage{amsfonts}
\usepackage{mathrsfs}
\usepackage{theorem}
\usepackage{color}
\usepackage{float}
\usepackage{subfig}
\DeclareMathAlphabet{\bm}{OML}{cmm}{b}{it}
\theorembodyfont{\rmfamily}
\newtheorem{theorem}{Theorem}
\newtheorem{lemma}{Lemma}

\newtheorem{corollary}{Corollary}
\newtheorem{remark}{Remark}
\newtheorem{proposition}{Proposition}

\newcommand{\qed}{\hfill \IEEEQED}

\newcommand{\bol}[1]{\mathbf{#1}}

\newcommand{\san}[1]{\mathsf{#1}}

\begin{document}

\title{Neyman-Pearson Test for Zero-Rate Multiterminal Hypothesis Testing\thanks{A part of this paper 
was presented in 2017 IEEE International Symposium on Information Theory.}}

\author{Shun Watanabe
\thanks{S.~Watanabe is with the Department of Computer and Information Sciences, Tokyo University of Agriculture and Technology, Japan, E-mail:shunwata@cc.tuat.ac.jp.}
}

% The paper headers
%\markboth{Journal of \LaTeX\ Class Files,~Vol.~6, No.~1, January~2007}%
%{Shell \MakeLowercase{\textit{et al.}}: Bare Demo of IEEEtran.cls for Journals}

\maketitle
\begin{abstract}
The problem of zero-rate multiterminal hypothesis testing is revisited from the perspective of
information-spectrum approach and finite blocklength analysis.
A Neyman-Pearson-like test is proposed and its non-asymptotic performance is clarified;
for a short block length, it is numerically determined that the proposed test is superior to the previously reported 
Hoeffding-like test proposed by Han-Kobayashi. 
For a large deviation regime, it is shown that our proposed test achieves an optimal trade-off
between the type I and type II exponents presented by Han-Kobayashi.
Among the class of symmetric (type based) testing schemes, when the type I error probability is non-vanishing, 
the proposed test is optimal up to the second-order term of the type II error exponent;
the latter term is characterized in terms of the variance of the projected relative entropy density.
The information geometry method plays an important role in the analysis as well as the construction
of the test.  
\end{abstract}

%%%%%%%%%%%%%%%%%%%%%%%%%%%%%%%%%%%%%%%%%%%
	\section{Introduction}
	
	In the classic hypothesis testing problem, upon observing $Z^n$, a tester tries to distinguish 
	whether the observation comes from the null hypothesis $P$ or the alternative hypothesis $Q$.
	It is widely known that the so-called Neyman-Pearson test \cite{NeyPea28} is the most powerful test in this regard,\footnote{Technically speaking, for finite blocklength,
	non-randomized Neyman-Pearson test is optimal only for limited number of trade-off points, and randomization is required in general (eg.~see \cite{Lehmann-Romano}).} 
	and the trade-off between the type I error probability $\alpha_n^\mathtt{NP}$ and the type II error probability $\beta_n^\mathtt{NP}$ is characterized as\footnote{Throughout the paper,
	we only consider discrete random variables taking values in finite sets.	 
	The notations $P(\cdot)$ and $Q(\cdot)$ in \eqref{eq:Neyman-Pearson} mean that 
	the probabilities of events are computed with respect to a sequence of i.i.d. (independent and identically distributed) random variables $Z^n$ that are
	distributed according to product distributions $P^n$ and $Q^n$, respectively. By a slight abuse of notation, we use the same notations $P$ and $Q$
	to describe the probability mass functions.}
	\begin{align} \label{eq:Neyman-Pearson}
	\alpha_n^\mathtt{NP} = P\bigg( \frac{1}{n} \sum_{i=1}^n \Lambda(Z_i) \le \tau \bigg),~~~\beta_n^\mathtt{NP} = Q\bigg( \frac{1}{n} \sum_{i=1}^n \Lambda(Z_i) > \tau \bigg).
	\end{align}
	Here, $\Lambda(z) = \imath_{P\|Q}(z) = \log \frac{P(z)}{Q(z)}$ is the log-likelihood ratio between the two distributions.
	This ratio is also known as the relative entropy density. 
	An application of the law of large numbers to \eqref{eq:Neyman-Pearson} implies that,
	for vanishing type I error probability, 
	the asymptotically optimal exponent of the type II error probability is given by the relative entropy $D(P\|Q)$;
	more refined analyses on \eqref{eq:Neyman-Pearson} give the tight bounds on more detailed asymptotics such as  the large deviation regime or
	the second-order regime \cite{blahut:74, Strassen:62}.

	Another important test, which we refer to as the Hoeffding test, entails comparing the type (empirical distribution) $\san{t}_{Z^n}$ of the observation
	with the null hypothesis \cite{Hoeffding65}; the null hypothesis is accepted if the relative entropy between
	the type and $P$ is smaller than a prescribed threshold, and is rejected otherwise.  The trade-off between 
	the type I error probability $\alpha_n^\mathtt{H}$ and the type II error probability $\beta_n^\mathtt{H}$ of this test is characterized as  
	\begin{align}
	\alpha_n^\mathtt{H} = P\bigg(   D(\san{t}_{Z^n} \| P) \ge r \bigg),~~~ \beta_n^\mathtt{H} = Q\bigg(  D(\san{t}_{Z^n} \| P) < r \bigg).
	\end{align}
	The advantage of this test is that it can be conducted without knowledge of $Q$,
	i.e., it is partially universal (see \cite{csiszar-korner:11}). Although the Hoeffding test delivers optimal performance asymptotically \cite[Theorem III.2]{csiszar:98},
	the trade-off between the errors is worse than that of the Neyman-Pearson test for a finite block length, in particular for a short block length.

	In \cite{berger:79b}, Berger introduced a new framework of multiterminal statistical decision problems
	under communication constraint. Inspired by his work, many researchers studied various problems of this kind 
	\cite{AhlCsi86, han:87, ZhaBer88, amari-han:89, amari:89, han-kobayashi:89, Amari:11}
	(see \cite{han-amari:98} for a thorough review). One important special case of these problems is the zero-rate multiterminal hypothesis testing problem, which is
	the main topic of this paper. This problem involves separate processing of the correlated observations $X^n$ and $Y^n$ by two encoders, after which messages are sent to
	a centralized decoder at zero-rate. Then, the decoder tries to distinguish whether the observations originate from the null hypothesis $P_{XY}$
	or the alternative hypothesis $Q_{XY}$. It was shown in \cite{han:87, ShaPap92} that,\footnote{More precisely, an achievability testing scheme was 
		proposed in \cite{han:87}, and the converse for the so-called one-bit compression case was shown under some regularity condition.
		Later, the converse for the zero-rate compression case was shown in \cite{ShaPap92} under some regularity condition. Recently, the regularity 
		condition of \cite{ShaPap92} was relaxed in \cite{ueta-kuzuoka:14}.
		Further, it is worth mentioning that the answer to
		the multiterminal hypothesis testing with a comparator is given by the same quantity, \eqref{eq:projected-ralative-entropy}, \cite{polyanskiy:12}.} 
	for a type I error with vanishing probability,
	the asymptotically optimal exponent of the type II error probability is given by the projected relative entropy defined by
	\begin{align} \label{eq:projected-ralative-entropy}
	E(P_{XY} \| Q_{XY}) = \min_{\tilde{P}_{XY} : \atop \tilde{P}_{X} = P_X, \tilde{P}_{Y}=P_Y} D(\tilde{P}_{XY} \| Q_{XY}).
	\end{align}
	In \cite{amari-han:89}, Amari-Han studied this problem from a differential geometrical viewpoint, and provided a geometrical interpretation of \eqref{eq:projected-ralative-entropy}
	by using the information geometry approach \cite{amari-nagaoka, Amari:book}.
	In fact, the term, projected relative entropy, should be clear from the observation in \cite{amari-han:89} (see Section \ref{sec:information-geometry}). 
	Furthermore, to study the large deviation regime of the zero-rate multiterminal hypothesis testing problem,
	Han-Kobayashi introduced a Hoeffding-like testing scheme for this problem \cite{han-kobayashi:89}; the trade-off in the error of their 
	testing scheme is characterized as 
	\begin{align} \label{eq:multi-hoeffding-type}
	\alpha_n^\mathtt{Hl} = P\bigg(   E(\san{t}_{X^n Y^n} \| P_{XY}) \ge r \bigg),~~~ \beta_n^\mathtt{Hl} = Q\bigg(  E(\san{t}_{X^nY^n} \| P_{XY}) < r \bigg),
	\end{align}
	where $\san{t}_{X^n Y^n}$ is the joint type of $(X^n,Y^n)$.
	It was shown in \cite{han-amari:98} that the bound in \eqref{eq:multi-hoeffding-type} is asymptotically tight in a large deviation regime.
	
	Thus far, we have reviewed the background on the zero-rate multiterminal hypothesis testing problem.  
	The main aim of this paper is to revisit this problem from the perspective of
	modern approaches developed in the past two decades, which are reviewed next.\footnote{At the time when the multiterminal hypothesis testing was actively studied in the late 80s to 
	early 90s, it seems that the method of type \cite{csiszar-korner:11} was the most popular tool for analysis.} 
	
	In their landmark paper \cite{han:93}, Han-Verd\'u proposed the information-spectrum approach. 
	Among other aspects, a key feature of the this approach is that the performance of a coding problem is
	characterized by the probabilistic behavior of the information density that is inherent to that coding problem.\footnote{Another key feature is its generality, i.e., 
	coding theorems are proved without any 
		assumptions such as stationarity and/or ergodicity.} 
	For instance, in the case of hypothesis testing, the relative entropy density $\imath_{P\|Q}$ can be regarded as the information density
	of this problem. The same philosophy was inherited by another recently popularized area of research, the finite block length 
	and the second-order analyses, pioneered by Hayashi \cite{hayashi:09} and Polyanskiy-Poor-Verd\'u \cite{polyanskiy:10}. 
	In the second-order analyses, instead of the law of the large number, the central limit theorem is applied to analyze 
	the probabilistic behavior of the information density up to $\sqrt{n}$ order. 
	
	Although information densities can be naturally introduced for
	some problems, it is non-trivial to identify the correct quantity in general. For instance, Kostina-Verd\'u introduced the $D$-tilted information density for the lossy source coding problem
	in a judicious manner, and characterized the second-order coding rate in terms of the variance of this information density \cite{kostina:12} (see also \cite{ingber:11} for
	an alternative approach proposed by Ingber-Kochman). 
	The same direction of research was extended to the Gray-Wyner network, one of the most basic multiterminal problems, by the author in \cite{Watanabe:15}
	(see also \cite{ZhoTanMot15}).
	
	As mentioned above, the classic hypothesis-testing problem comprises two important tests: the Neyman-Pearson 
	and Hoeffding tests. The test proposed by Han-Kobayashi \cite{han-kobayashi:89} can be regarded as a Hoeffding test for 
	the zero-rate multiterminal hypothesis testing. Thus, it is tempting, both theoretically and practically, to have
	a testing scheme that is reminiscent of the Neyman-Pearson test.
	In this paper, we propose such a testing scheme.
	In fact, the trade-off between the type I and type II error probabilities by our proposed test has the following form:
	\begin{align}
	\alpha_n^\mathtt{NPl} = P\bigg( \frac{1}{n} \sum_{i=1}^n \Lambda_\lambda(X_i,Y_i) \le \tau \bigg),~~~\beta_n^\mathtt{NPl} = Q\bigg( \frac{1}{n} \sum_{i=1}^n \Lambda_\lambda(X_i,Y_i) > \tau \bigg).
	\end{align}
	Here, $\Lambda_\lambda(x,y)$ is a proxy of the log-likelihood ratio parametrized by
	$\lambda \in [-E(Q_{XY}\|P_{XY}), E(P_{XY}\|Q_{XY})]$; as is subsequently shown, identification of $\Lambda_\lambda(x,y)$ is non-trivial, which is one of
	technical contributions of this paper.
	In contrast to the Neyman-Pearson test in the classic hypothesis testing, the proxy of the log-likelihood is parametrized by $\lambda$.
	As we will see later in the paper, adjustment of $\lambda$ depending on threshold $\tau$ is very important. 
	For instance, the optimal choice
	turns out to be $\lambda = \tau$ in the large deviation regime.\footnote{As a rule of thumb, $\lambda = \tau$ provides the best error trade-off
	even in the finite blocklength regime. }
	An extreme case $\lambda=E(P_{XY}\|Q_{XY})$ of this proxy of the log-likelihood ratio,
	which we term the projected relative entropy density, is $\jmath_{P\|Q}(x,y) = \log \frac{P_{XY}^*(x,y)}{Q_{XY}(x,y)}$ for the optimizer $P_{XY}^*$ 
	of $E(P_{XY}\|Q_{XY})$ in \eqref{eq:projected-ralative-entropy}. In fact, it will be clarified that the expectation of $\jmath_{P\|Q}(X,Y)$ over $P_{XY}$ is given by
	$E(P_{XY} \| Q_{XY})$. 
	
	Although it is not clear whether our proposed testing scheme is the most powerful,
	for a rather short block length, we numerically determine that our proposed testing scheme has better error trade-off than that of 
	the previously known test of Han-Kobayashi. 
	We also show that, for a large deviation regime, our proposed test achieves optimal trade-off
	between the type I and type II exponents shown by Han-Kobayashi.
	Furthermore, among the class of symmetric (type based) testing schemes, we 
	derive the optimal second-order rate of the type II exponent, which can be achieved by our proposed test. 
	Ultimately, it emerges that the optimal second-order rate is characterized by the variance of the projected relative entropy density. 
	
	Here, we would like to mention the similarity and dissimilarity of the Neyman-Pearson test in
	classic hypothesis testing and our Neyman-Pearson-like test in the multiterminal hypothesis testing
	from a geometrical viewpoint.  
	In terms of the classic hypothesis-testing problem, the Neyman-Pearson test is known to correspond to
	bisecting the probability simplex by a mixture family generated by the log-likelihood ratio $\Lambda(z)$,
	which is orthogonal to the e-geodesic connecting the null hypothesis $P$
	and the alternative hypothesis $Q$ (e.g.,~see \cite{nakagawa:93b}). 
	On the other hand, our Neyman-Pearson-like test of the multiterminal hypothesis testing bisects 
	the probability simplex by a mixture family generated by the proxy $\Lambda_\lambda(x,y)$ of the log-likelihood ratio.
	In contrast to the Neyman-Pearson test in classic hypothesis testing, our Neyman-Pearson-like test has the freedom to
	adjust the direction of bisection by parameter $\lambda$. Interestingly, this adjustment of direction is crucial to achieve the 
	optimal trade-off between the type I and type II error exponents in a large deviation regime. 
	
	In addition to the above mentioned geometrical motivation,
	there is also a practical motivation to introduce the Neyman-Pearson-like test. 
	Even for multiterminal hypothesis testing, we can conduct the Neyman-Pearson
	test by computing the log-likelihood ratio for a given observation. However, computing such a log-likelihood ratio
	is intractable as the blocklength become larger. On the other hand, our Neyman-Pearson-like
	scheme only requires computing the empirical average of the proxy of the log-likelihood ratio, and it is easier to implement
	(see Remark \ref{remark:trivial-test}).
	
	The remainder of the paper is organized as follows:
	In Section \ref{sec:problem-formulation}, we introduce our notation, and recall the problem formulation
	of the multiterminal hypothesis testing. 
	We also review the previously known testing scheme of Han-Kobayashi. 
	In Section \ref{sec:information-geometry}, we review basic results on information geometry
	as well as the result of Amari-Han \cite{amari-han:89}. 
	In Section \ref{section:non-asymptotic}, 
	we introduce our novel testing scheme.
	In Section \ref{sec:example}, we consider a binary example, and compare the two testing schemes. 
	In Section \ref{sec:LDP}, the large deviation performance of our proposed test is discussed.  
	In Section \ref{sec:second-order}, we derive the second-order exponent.
	We conclude the paper with some discussions in Section \ref{sec:conclusion}.

	%%%%%%%%%%%%%%%%%%%%%%%%%%%%%%%%%%%%%%%%%%%
	\section{Problem Formulation}
	\label{sec:problem-formulation}

	\subsection{Notation}
	
	Random variables (e.g.,~$X$) and their realizations (e.g.,~$x$) are presented in upper and lower 
	case format, respectively. All random variables take values in some finite alphabets, which are
	denoted in calligraphic font (e.g.,~${\cal X}$). The cardinality of ${\cal X}$ is denoted as $|{\cal X}|$.
	Let the random vector $X^n = (X_1,\ldots,X_n)$ and similarly for its realization 
	$\bm{x} = (x_1,\ldots,x_n)$. For information theoretic quantities, 
	we follow the same notations as \cite{csiszar-korner:11}; e.g.,~the 
	entropy and the relative entropy are denoted by $H(X)$ and $D(P\|Q)$, respectively.
	For two distributions $P$ and $P^\prime$, the variational distance is denoted by $\| P - P^\prime \|$.
	For a sequence $\bm{x}$, its type (empirical distribution) is denoted by $\san{t}_{\bm{x}}$; similarly the joint type 
	of $(\bm{x},\bm{y})$ is denoted by $\san{t}_{\bm{x}\bm{y}}$.
	The set of all positive probability distributions on ${\cal X}$ is denoted by ${\cal P}({\cal X})$, 
	the set of all probability distributions (not necessarily positive) is denoted by $\overline{{\cal P}}({\cal X})$, and the set of all
	types (not necessarily positive) on ${\cal X}$ is denoted by ${\cal P}_n({\cal X})$. Similar notations are used for the sets of 
	joint distributions and joint types.
	In addition, the notations $\san{E}[\cdot]$ and $\san{Var}[\cdot]$ mean computing the expectation and the variance of
	the random variables in the bracket, respectively.
	$\Phi(t) = \int_{-\infty}^t \frac{1}{\sqrt{2\pi}} e^{- \frac{u^2}{2}} du$ is the cumulative distribution function of the standard normal
	distribution; its inverse is denoted by $\Phi^{-1}(\varepsilon)$ for $0 < \varepsilon < 1$. 
	Throughout the paper, the base of $\log$ and $\exp$ are $e$. 
	%Further notations on information geometry will be introduced in Section \ref{sec:information-geometry}. 
	
	%%
	\subsection{Multiterminal Hypothesis Testing}
	
	In this section, we introduce the problem setting and review some basic results. 
	We consider the statistical problem of testing the null hypothesis $\san{H}_0:P_{XY}$ on ${\cal X} \times {\cal Y}$
	versus the alternative hypothesis $\san{H}_1:Q_{XY}$ on the same alphabet. We assume $P_{XY}$ and $Q_{XY}$ are positive, i.e.,
	they have full support throughout the paper.\footnote{The results on multiterminal hypothesis testing \cite{ShaPap92, han-amari:98}
	are sensitive to the supports of the alternative hypothesis. Even though the full support assumption may be slightly relaxed (eg.~see \cite{ueta-kuzuoka:14}),
	we only consider the full support case in this paper.} The i.i.d. random variables $(X^n,Y^n)$ distributed according to either $P_{XY}^n$ or
	$Q_{XY}^n$ are observed separately by two terminals, and they are encoded by two encoders
	\begin{align}
	& f_1^{(n)}:{\cal X}^n \to {\cal M}_1^{(n)}, \\
	& f_2^{(n)}:{\cal Y}^n \to {\cal M}_2^{(n)},
	\end{align}
	respectively. Then, the decoder
	\begin{align}
	g^{(n)}:{\cal M}_1^{(n)} \times {\cal M}_2^{(n)} \to \{\san{H}_0,\san{H}_1\}
	\end{align}
	decides whether to accept the null hypothesis. When the block length $n$ is obvious from the context, we omit the superscript $n$.
	For a given testing scheme $T_n = (f_1,f_2,g)$, the type I error probability is defined by
	\begin{align}
	\alpha[T_n] := P\bigg( g(f_1(X^n),f_2(Y^n)) = \san{H}_1 \bigg)
	\end{align}
	and the type II error probability is defined by
	\begin{align}
	\beta[T_n] := Q\bigg( g(f_1(X^n), f_2(Y^n)) = \san{H}_0 \bigg).
	\end{align}
	In the remaining part of the paper, $P(\cdot)$ (or $Q(\cdot)$) means $(X^n,Y^n)$ is distributed 
	according to $P_{XY}^n$ (or $Q_{XY}^n$).
	
	A sequence of the testing scheme $\{ T_n \}_{n=1}^\infty$ is said to be zero-rate if 
	\begin{align}
	\lim_{n\to \infty} \frac{1}{n} \log |{\cal M}_i^{(n)} | =0,~i=1,2.
	\end{align}
	Furthermore, $\{T_n\}_{n=1}^\infty$ is said to be a symmetric testing scheme if the encoders
	$f_1^{(n)}$ and $f_2^{(n)}$ only depend on the marginal types of $\bm{x}$ and $\bm{y}$, 
	respectively. The class of zero-rate schemes was introduced in \cite{han:87} and
	the class of symmetric schemes was introduced in \cite{amari-han:89}.
	Note that a symmetric scheme is automatically zero-rate since the number of marginal types 
	is polynomial in $n$. 
	
	Let 
	\begin{align} \label{eq:definition-E0}
	E_0(P_{XY} \|Q_{XY}) := \sup\left\{ \liminf_{n\to\infty} - \frac{1}{n}\log \beta[T_n]: \{T_n\}_{n=1}^\infty \mbox{ is zero-rate},\lim_{n\to\infty} \alpha[T_n]=0 \right\}
	\end{align}
	and
	\begin{align} \label{eq:definition-Es}
	E_\san{s}(P_{XY}\|Q_{XY}) := \sup\left\{ \liminf_{n\to\infty} - \frac{1}{n}\log \beta[T_n]: \{T_n\}_{n=1}^\infty \mbox{ is symmetric},\lim_{n\to\infty} \alpha[T_n]=0 \right\}
	\end{align}
	be the optimal exponent of the type II error probability in each class of schemes.
	%\footnote{Although $E_0(P_{XY})$ and $E_\san{s}(P_{XY})$ depend on $Q_{XY}$,
	%we omit it from the notations since $Q_{XY}$ is fixed throughout the paper.} 
	By definition, $E_0(P_{XY}\|Q_{XY}) \ge E_\san{s}(P_{XY}\|Q_{XY})$.
	These quantities can be characterized as follows.
	\begin{proposition}[\cite{han:87, ShaPap92}] \label{proposition:first-order}
		It holds that\footnote{The exponents defined in \eqref{eq:definition-E0} and \eqref{eq:definition-Es} are the so-called weak converse, i.e., we require 
		that the type I error probability converges to $0$. In fact, it is known that the strong converse holds for this problem \cite{ShaPap92}.} 
		\begin{align}
		E_0(P_{XY}\|Q_{XY}) = E_\san{s}(P_{XY}\|Q_{XY}) = E(P_{XY}\|Q_{XY}),
		\end{align}
		where
		\begin{align} \label{eq:optimal-exponent-expression}
		E(P_{XY}\|Q_{XY}) := \min\left\{ D(\tilde{P}_{XY} \| Q_{XY}) : \tilde{P}_{XY} \in \overline{{\cal P}}({\cal X}\times{\cal Y}), \tilde{P}_X=P_X, \tilde{P}_Y=P_Y \right\}.
		\end{align}
	\end{proposition}

	%%%%%%%%%%%
	\subsection{Han-Kobayashi Testing Scheme} \label{subsec:han-kobayashi-scheme}
	
	In \cite{han-kobayashi:89}, Han-Kobayashi studied a large deviation regime of multiterminal hypothesis testing.
	For $0 \le r \le E(Q_{XY} \| P_{XY})$, let 
	\begin{align}
	F(r) := \sup\left\{ \liminf_{n\to\infty} - \frac{1}{n} \log \beta[T_n] : \{T_n\}_{n=1}^\infty \mbox{ is zero-rate}, \liminf_{n\to\infty} -\frac{1}{n} \log \alpha[T_n] \ge r \right\}.
	\end{align}
	To derive a lower bound on $F(r)$, Han-Kobayashi proposed the following  Hoeffding-like testing scheme.\footnote{The testing scheme proposed in \cite{han-kobayashi:89}
		is slightly different, but it is essentially the same as the scheme reviewed in this section.} By definition, note that 
	\begin{align} \label{eq:identity-E-relative-entropy}
	E(P_{\bar{X}\bar{Y}} \| P_{XY}) = E(P_{\bar{X}} \times P_{\bar{Y}} \| P_{XY})
	\end{align}
	holds for any joint distribution $P_{\bar{X}\bar{Y}}$, where $P_{\bar{X}} \times P_{\bar{Y}}$ are the product distribution 
	of the marginals $P_{\bar{X}}, P_{\bar{Y}}$ of $P_{\bar{X}\bar{Y}}$.
	Upon observing $\bm{x}$ and $\bm{y}$, the encoders send their types.  Then, upon receiving a pair of marginal types $(\san{t}_{\bm{x}},\san{t}_{\bm{y}})$,
	the decoder computes $E(\san{t}_{\bm{x}} \times \san{t}_{\bm{y}} \| P_{XY})$; if the value is smaller than a prescribed threshold $r$, then it outputs
	$\san{H}_0$; otherwise, it outputs $\san{H}_1$. By \eqref{eq:identity-E-relative-entropy}, $g(f_1(\bm{x}),f_2(\bm{y})) = \san{H}_0$ if and only if
	\begin{align}
	E(\san{t}_{\bm{x}\bm{y}} \| P_{XY}) < r.
	\end{align}
	In fact, the threshold $r$ controls the convergence speed of the type I error probability, i.e., in the large deviation regime, 
	the type I error probability behaves as $\exp\{-nr\}$. Apparently, this scheme is a symmetric scheme.
	The performance of this scheme is summarized in the following proposition.

	\begin{proposition} \label{proposition:performance-han-kobayashi}
		For a given $r > 0$, the above mentioned Hoeffding like scheme $T_n^\mathtt{Hl}$ has the following error trade-off: 
		\begin{align}
		\alpha[T_n^\mathtt{Hl}] &= P\bigg( E(\san{t}_{X^n Y^n} \| P_{XY}) \ge r \bigg), \\
		\beta[T_n^\mathtt{Hl}] &= Q\bigg( E(\san{t}_{X^n Y^n} \| P_{XY}) < r \bigg).
		\end{align}
	\end{proposition}
	
	It was shown in \cite{han-kobayashi:89} that the above testing scheme satisfies 
	\begin{align} 
	\lim_{n\to \infty} - \frac{1}{n} \log \alpha[T_n^\mathtt{Hl}] = r
	\end{align}
	and
	\begin{align} \label{eq:optimal-LDP}
	\lim_{n\to \infty} - \frac{1}{n} \log \beta[T_n^\mathtt{Hl}] 
	&= \min\big\{ D(P_{\bar{X}\bar{Y}} \| Q_{XY}) : E(P_{\bar{X}\bar{Y}} \| P_{XY}) \le r \big\} \\
	&= \min\big\{ E(P_{\bar{X} \bar{Y}} \| Q_{XY}) : D(P_{\bar{X}\bar{Y}} \| P_{XY}) \le r \big\},
	\end{align}
	which is optimal among the class of all zero-rate testing schemes \cite{han-amari:98}; we summarize these 
	results in the following proposition.
	
	\begin{proposition}[\cite{han-kobayashi:89, han-amari:98}]
		For $0 \le r \le E(Q_{XY} \| P_{XY})$, it holds that\footnote{In fact, only the expression \eqref{eq:exponent-expression-1}
		for $F(r)$ was derived in \cite{han-kobayashi:89, han-amari:98}; however, the expression \eqref{eq:exponent-expression-1} 
		can be also described by the expression \eqref{eq:exponent-expression-2}. 
		Later in Lemma \ref{lemma:properties-optimizer}, we will verify that the two expressions coincide for $0 < r < E(Q_{XY} \| P_{XY})$; for
		$r=0$ and $r = E(Q_{XY} \| P_{XY})$, it is not difficult to see that the two expressions coincide and are given by $E(P_{XY}\|Q_{XY})$ and $0$, respectively.} 
		\begin{align}
		F(r) &= \min\big\{ D(Q_{\bar{X}\bar{Y}} \| Q_{XY}) : E(Q_{\bar{X}\bar{Y}} \| P_{XY}) \le r \big\} \label{eq:exponent-expression-1} \\
		&= \min\big\{ E(P_{\bar{X} \bar{Y}} \| Q_{XY}) : D(P_{\bar{X}\bar{Y}} \| P_{XY}) \le r \big\}. \label{eq:exponent-expression-2}
		\end{align}
	\end{proposition}

	%%%%%%%%%%%%%%%%%%%%%%%%%%%%%%%%%%%%%%%%%%%
	\section{Preliminaries of Information Geometry}
	\label{sec:information-geometry}
	
	In this section, we review some results on information geometry that are needed in later sections.
	Interested readers are referred to \cite{amari-nagaoka, csiszar-shields:04book} for a thorough review on information geometry.
	
	%%%%%%%%%%%%%%%
	\subsection{Properties of Projection}
	
	In this section, we review some properties of projection with respect to the relative entropy, which is sometimes known as $I$-projection.
	Let ${\cal C} \subseteq \overline{{\cal P}}({\cal Z})$ be a (nonempty) closed convex set. For a given $Q \in {\cal P}({\cal Z})$, let us 
	consider the following optimization problem:
	\begin{align} \label{eq:convex-optimization}
	\min_{\tilde{P} \in {\cal C}} D(\tilde{P} \| Q).
	\end{align}
	The optimizer of \eqref{eq:convex-optimization} satisfies the following extremal condition \cite{csiszar:75} (see also \cite[Section 3]{csiszar-shields:04book}).
	%%%
	\begin{theorem}[\cite{csiszar:75}] \label{theorem:extremal}
		The optimizer $P^\star$ of \eqref{eq:convex-optimization} is unique; furthermore, for every $\tilde{P} \in {\cal C}$, the optimizer $P^\star$ satisfies
		$\mathtt{supp}(\tilde{P}) \subseteq \mathtt{supp}(P^\star)$ and
		\begin{align} \label{eq:tangent}
		D(\tilde{P} \| Q) \ge D( \tilde{P} \| P^\star) + D(P^\star \| Q).
		\end{align}
	\end{theorem}
	
	Note that \eqref{eq:tangent} is equivalent to
	\begin{align} \label{eq:tangent-2}
	\sum_z (\tilde{P}(z) - P^\star(z)) \log \frac{P^\star(z)}{Q(z)} \ge 0.
	\end{align}
	The set 
	\begin{align} \label{eq:tangent-plane-single-terminal}
	{\cal S}_=(P^\star) := \bigg\{ \tilde{P} : \sum_z (\tilde{P}(z) - P^\star(z)) \log \frac{P^\star(z)}{Q(z)} = 0 \bigg\}
	\end{align}
	is the tangent plane of ${\cal C}$ at $P^\star$ in the sense of \eqref{eq:tangent-2}. 
	Moreover, we have 
	\begin{align} \label{eq:inclusion-convex-supporting}
	{\cal C} \subseteq {\cal S}_{\ge}(P^\star) := \bigg\{ \tilde{P} : \sum_z (\tilde{P}(z) - P^\star(z)) \log \frac{P^\star(z)}{Q(z)} \ge 0 \bigg\}
	\end{align}
	and
	\begin{align}
	\min_{\tilde{P} \in {\cal S}_{\ge}(P^\star) } D(\tilde{P} \| Q) = D(P^\star \| Q).
	\end{align}
	
	For given functions $f_1,\ldots,f_k$ from ${\cal Z}$ to $\mathbb{R}$ and constants $c_1,\ldots,c_k \in \mathbb{R}$, the set
	\begin{align}
	{\cal M} := \bigg\{ \tilde{P} \in \overline{{\cal P}}({\cal Z}): \sum_z \tilde{P}(z) f_i(z) = c_i,~\forall 1 \le i \le k \bigg\}
	\end{align}
	is known as a mixture family.\footnote{Sometimes, it is also referred to as a linear family.} We make the assumption that ${\cal M}$ is not empty, and that it contains at least one
	element $\tilde{P}$ having full support. 
	On the other hand, the set ${\cal E}$ of all distributions of the form
	\begin{align}
	\tilde{P}(z) = Q(z) \exp\bigg[ \sum_{i=1}^k \theta_i f_i(z) - \psi(\theta) \bigg],~~~\theta = (\theta_1,\ldots,\theta_k) \in \mathbb{R}^k
	\end{align}
	is termed the exponential family
	generated by $Q$ and $f_1,\ldots,f_k$; the normalization constant $\psi(\theta)$ is usually known as a potential function.
	When ${\cal C} = {\cal M}$, the optimizer of \eqref{eq:convex-optimization} satisfies the following Pythagorean identity,
	and the optimizer is included in the exponential family ${\cal E}$.
	
	\begin{theorem}[\cite{csiszar:75}] \label{theorem:pythagorean-linear-family}
		When ${\cal C} = {\cal M}$ is a mixture family, the optimizer $P^\star$ of \eqref{eq:convex-optimization} satisfies 
		\begin{align} \label{eq:pythagorean-indentity-single-variable}
		D(\tilde{P} \| Q) = D( \tilde{P} \| P^\star) + D(P^\star \| Q)
		\end{align}
		for every $\tilde{P} \in {\cal M}$, and $P^\star \in {\cal E} \cap {\cal M}$.
		%Conversely, the intersection $P^\star \in {\cal E} \cap {\cal L}$ is unique, and it satisfies \eqref{eq:pythagorean-indentity-single-variable}
		%for every $\tilde{P} \in {\cal L}$.
	\end{theorem}

	%%%%%%%%%%%%%%%
	\subsection{Geometry of ${\cal P}({\cal X}\times {\cal Y})$}
	
	In this section, we review the results in \cite{amari-han:89} (see also \cite{amari:01, amari-nagaoka}). We first introduce a coordinate system on the set of all
	positive joint distributions, ${\cal P}({\cal X}\times {\cal Y})$. 
	Note that ${\cal P}({\cal X}\times {\cal Y})$ is a $(|{\cal X}| |{\cal Y}| -1)$-dimensional manifold. 
	When we consider a multiterminal problem, 
	it is convenient to consider the following parametrization specified by $(d_{\san{x}}d_{\san{y}}+d_{\san{x}}+d_{\san{y}})$ parameters, where 
	$d_{\san{x}} = |{\cal X}|-1$ and $d_{\san{y}} = |{\cal Y}|-1$. Henceforth, we identify the alphabets as 
	${\cal X} = \{0,1,\ldots,d_{\san{x}} \}$ and ${\cal Y} = \{0,1,\ldots,d_{\san{y}}\}$.
	First, a natural parameter is introduced as 
	\begin{align} 
	P_{XY,\theta}(x,y) := \exp\bigg[ \sum_{i=1}^{d_{\san{x}}} \theta_i^{\san{x}} \delta_i(x) + \sum_{j=1}^{d_{\san{y}}} \theta_j^{\san{y}} \delta_j(y) 
	+ \sum_{i=1}^{d_{\san{x}}} \sum_{j=1}^{d_{\san{x}}} \theta_{ij}^{\san{xy}} \delta_{ij}(x,y) - \psi(\theta) \bigg],
	\label{eq:natural-parameterization}
	\end{align}
	where $\delta_i(x) = \bol{1}[x=i]$ is the Kronecker delta ($\delta_j(y)$ and $\delta_{ij}(x,y)$ are defined similarly), 
	\begin{align}
	\theta_i^{\san{x}} &= \log \frac{P_{XY,\theta}(i,0)}{P_{XY,\theta}(0,0)}, \\
	\theta_j^{\san{y}} &= \log \frac{P_{XY,\theta}(0,j)}{P_{XY,\theta}(0,0)}, \\
	\theta_{ij}^{\san{xy}} &= \log \frac{P_{XY,\theta}(i,j) P_{XY,\theta}(0,0)}{P_{XY,\theta}(i,0)P_{XY,\theta}(0,j)} \label{eq:correlation-coordinate}
	\end{align}
	for $1 \le i \le d_{\san{x}}$ and $1 \le j \le d_{\san{y}}$, 
	and $\psi(\theta)$ is the normalization constant  (potential function) given by 
	\begin{align}
	\psi(\theta) &= - \log P_{XY,\theta}(0,0) \\
	&= \log \bigg( 1 + \sum_{i=1}^{d_{\san{x}}} e^{\theta_i^{\san{x}}} + \sum_{j=1}^{d_{\san{y}}} e^{\theta_j^{\san{y}}} 
	+ \sum_{i=1}^{d_{\san{x}}} \sum_{j=1}^{d_{\san{y}}} e^{\theta_i^{\san{x}} + \theta_j^{\san{y}} + \theta_{ij}^{\san{xy}} } \bigg).
	\label{eq:potential-function-2}
	\end{align}
	In this coordinate system, $\{ \theta_{ij}^{\san{xy}} \}$ describe the correlation between $X$ and $Y$. In fact, we can verify that
	$P_{XY,\theta} = P_{X,\theta} \times P_{Y,\theta}$ if and only if $\theta_{ij}^{\san{xy}} = 0$ for every $1 \le i \le d_{\san{x}}$ and $1 \le j \le d_{\san{y}}$.
	
	Next, the expectation parameter is introduced as 
	\begin{align}
	P_{XY,\eta}(x,y) &:= \sum_{i=1}^{d_{\san{x}}} \eta_i^{\san{x}}\big( \delta_i(x) - \delta_0(x) \big) \delta_0(y) + \sum_{j=1}^{d_{\san{y}}} \eta_j^{\san{y}} \delta_0(x) \big( \delta_j(y) - \delta_0(y) \big) \\
	&~~~+ \sum_{i=1}^{d_{\san{x}}} \sum_{j=1}^{d_{\san{y}}} \eta_{ij}^{\san{xy}} \big( \delta_{ij}(x,y) - \delta_i(x) \delta_0(y) - \delta_0(x) \delta_j(y) + \delta_0(x) \delta_0(y) \big) + \delta_0(x) \delta_0(y),
	\label{eq:expansion-of-expectation-parameter}
	\end{align}
	where 
	\begin{align}
	\eta_i^{\san{x}} &= P_{X,\eta}(i), \label{eq:expectation-parameter-1} \\
	\eta_j^{\san{y}} &= P_{Y,\eta}(j), \label{eq:expectation-parameter-2} \\
	\eta_{ij}^{\san{xy}} &= P_{XY,\eta}(i,j) \label{eq:expectation-parameter-3}
	\end{align}
	for $1 \le i \le d_{\san{x}}$ and $1 \le j \le d_{\san{y}}$. 
	The parameters \eqref{eq:expectation-parameter-1}-\eqref{eq:expectation-parameter-3} are called expectation parameters since
	$\eta_i^\san{x} = \san{E}[\delta_i(X)]$, $\eta_j^\san{y} = \san{E}[\delta_j(X)]$, and $\eta_{ij}^\san{xy} = \san{E}[\delta_{ij}(X,Y)]$ for $(X,Y) \sim P_{XY,\eta}$.
	Apparently in this coordinate system, $\{ \eta_i^{\san{x}} \}$ and $\{ \eta_j^{\san{y}} \}$ describe 
	the marginal distributions, respectively. 
	
	By taking the derivative of \eqref{eq:potential-function-2}, we have
	\begin{align}
	\frac{\partial \psi(\theta)}{\partial \theta_i^{\san{x}}} &= \exp\big[ \theta_i^{\san{x}} - \psi(\theta) \big] 
	+ \sum_{j=1}^{d_{\san{y}}} \exp\big[ \theta_i^{\san{x}} + \theta_j^{\san{y}} + \theta_{ij}^{\san{xy}} - \psi(\theta) \big] \\
	&= \sum_{j=0}^{d_{\san{y}}} P_{XY,\theta}(i,j) \\
	&= P_{X,\theta}(i)
	\end{align}
	for $1 \le i \le d_{\san{x}}$.
	Similarly, we have
	\begin{align}
	\frac{\partial \psi(\theta)}{\partial \theta_j^{\san{y}}} = P_{Y,\theta}(j)
	\end{align}
	for $1 \le j \le d_{\san{y}}$, and 
	\begin{align}
	\frac{\partial \psi(\theta)}{\partial \theta_{ij}^{\san{xy}}} &= \exp\big[ \theta_i^{\san{x}} + \theta_j^{\san{y}} + \theta_{ij}^{\san{xy}} - \psi(\theta) \big] \\
	&= P_{XY,\theta}(i,j)
	\end{align}
	for $1 \le i \le d_{\san{x}}$ and $1 \le j \le d_{\san{y}}$. Thus, the two coordinate systems $\theta$
	and $\eta$ are related by
	\begin{align}
	\eta_i^{\san{x}}(\theta) &= \frac{\partial \psi(\theta)}{\partial \theta_i^{\san{x}}}, \\
	\eta_j^{\san{y}}(\theta) &= \frac{\partial \psi(\theta)}{\partial \theta_j^{\san{y}}}, \\
	\eta_{ij}^{\san{xy}}(\theta) &= \frac{\partial \psi(\theta)}{\partial \theta_{ij}^{\san{xy}}}.
	\end{align}
	Hereafter, we use notations such as $\theta^{\san{xy}} = (\theta_{ij}^{\san{xy}} : 1\le i \le d_{\san{x}},1\le j \le d_{\san{y}})$ and
	$\eta^{\san{y}} = (\eta_j^{\san{y}} : 1 \le j \le d_{\san{y}})$  for brevity.
	
	Let $\theta(P)$ and $\eta(P)$ be the natural and expectation parameters that correspond to $P_{XY}$,  
	and let $\theta(Q)$ and $\eta(Q)$ be the natural and expectation parameters that correspond to $Q_{XY}$.
	Let 
	\begin{align}
	{\cal E}(\theta^{\san{xy}}(Q)) := \big\{ P_{XY,\theta} : \theta^{\san{xy}} = \theta^{\san{xy}}(Q) \big\}
	\end{align} 
	be the exponential family containing $Q_{XY}$, and let 
	\begin{align}
	{\cal M}(\eta^{\san{x}}(P),\eta^{\san{y}}(P)) := \big\{ P_{XY,\theta} : \eta^{\san{x}}(\theta) = \eta^{\san{x}}(P), \eta^{\san{y}}(\theta) = \eta^{\san{y}}(P) \big\}
	\end{align}
	be the mixture family containing $P_{XY}$. Similarly, we define ${\cal E}(\theta^{\san{xy}}(P))$ and 
	${\cal M}(\eta^{\san{x}}(Q),\eta^{\san{y}}(Q))$ by replacing the roles of $P_{XY}$ and $Q_{XY}$.
	It was shown in \cite{amari-han:89} that the optimization problem in $E(P_{XY} \| Q_{XY})$
	is achieved by the intersection of ${\cal E}(\theta^{\san{xy}}(Q))$ and ${\cal M}(\eta^{\san{x}}(P),\eta^{\san{y}}(P))$, 
	and the interpretation of its Pythagorean theorem was given (see Fig.~\ref{Fig:pythagorean}).\footnote{In fact,
		Theorem \ref{theorem:pythagorean} is essentially a special case of Theorem \ref{theorem:pythagorean-linear-family}. We reviewed 
		both of these claims for later convenience.} 

	\begin{theorem}[\cite{amari-han:89}] \label{theorem:pythagorean}
		The optimizer $P_{XY}^*$ of $E(P_{XY}\|Q_{XY})$
		satisfies $P_{XY}^* \in {\cal E}(\theta^{\san{xy}}(Q)) \cap {\cal M}(\eta^{\san{x}}(P),\eta^{\san{y}}(P))$ and
		\begin{align} \label{eq:pythagorean}
		D(P_{XY} \| Q_{XY}) = D(P_{XY} \| P^*_{XY}) + D(P^*_{XY} \| Q_{XY}).
		\end{align}
		Similarly, the optimizer $Q_{XY}^*$ of $E(Q_{XY} \| P_{XY})$ satisfies $Q_{XY}^* \in {\cal E}(\theta^{\san{xy}}(P)) \cap {\cal M}(\eta^{\san{x}}(Q),\eta^{\san{y}}(Q))$ and
		\begin{align} \label{eq:pythagorean-2}
		D(Q_{XY} \| P_{XY}) = D(Q_{XY} \| Q_{XY}^*) + D(Q_{XY}^* \| P_{XY}).
		\end{align}
	\end{theorem}
	
	Theorem \ref{theorem:pythagorean} implies that $E(P_{XY} \| Q_{XY})$ is a projected component of
	the relative entropy $D(P_{XY} \| Q_{XY})$; thus, we term it the projected relative entropy.
	
	For later use, we introduce a simple implication of Theorem \ref{theorem:pythagorean}.
	%%%%%%
	\begin{corollary} \label{corollary:implication-pythagorean}
		For any $\tilde{P}_{XY} \in {\cal E}(\theta^{\san{xy}}(P))$ and $\tilde{Q}_{XY} \in {\cal E}(\theta^{\san{xy}}(Q))$ 
		such that $\tilde{P}_X = \tilde{Q}_X$ and $\tilde{P}_Y = \tilde{Q}_Y$, it holds that
		\begin{align}
		\sum_{x,y} \tilde{P}_{XY}(x,y) \log \frac{\tilde{Q}_{XY}(x,y)}{Q_{XY}(x,y)} = D(\tilde{Q}_{XY} \| Q_{XY}).
		\end{align}
	\end{corollary}
	%%%
	\begin{proof}
		By applying Theorem \ref{theorem:pythagorean} for $(\tilde{P}_{XY},\tilde{Q}_{XY})$ in the place of $(P_{XY},P_{XY}^*)$, we have
		\begin{align}
		D(\tilde{Q}_{XY} \| Q_{XY}) &= D(\tilde{P}_{XY} \| Q_{XY}) - D(\tilde{P}_{XY} \| \tilde{Q}_{XY}) \\
		&= \sum_{x,y} \tilde{P}_{XY}(x,y) \bigg[ \log \frac{\tilde{P}_{XY}(x,y)}{Q_{XY}(x,y)} - \log \frac{\tilde{P}_{XY}(x,y)}{\tilde{Q}_{XY}(x,y)} \bigg] \\
		&= \sum_{x,y} \tilde{P}_{XY}(x,y) \log \frac{\tilde{Q}_{XY}(x,y)}{Q_{XY}(x,y)}. 
		\end{align}
	\end{proof}
	
	%%%
	%%%
\begin{figure}[tb]
\centering{
%\begin{minipage}{.4\textwidth}
\includegraphics[width=0.4\textwidth, bb=0 0 350 312]{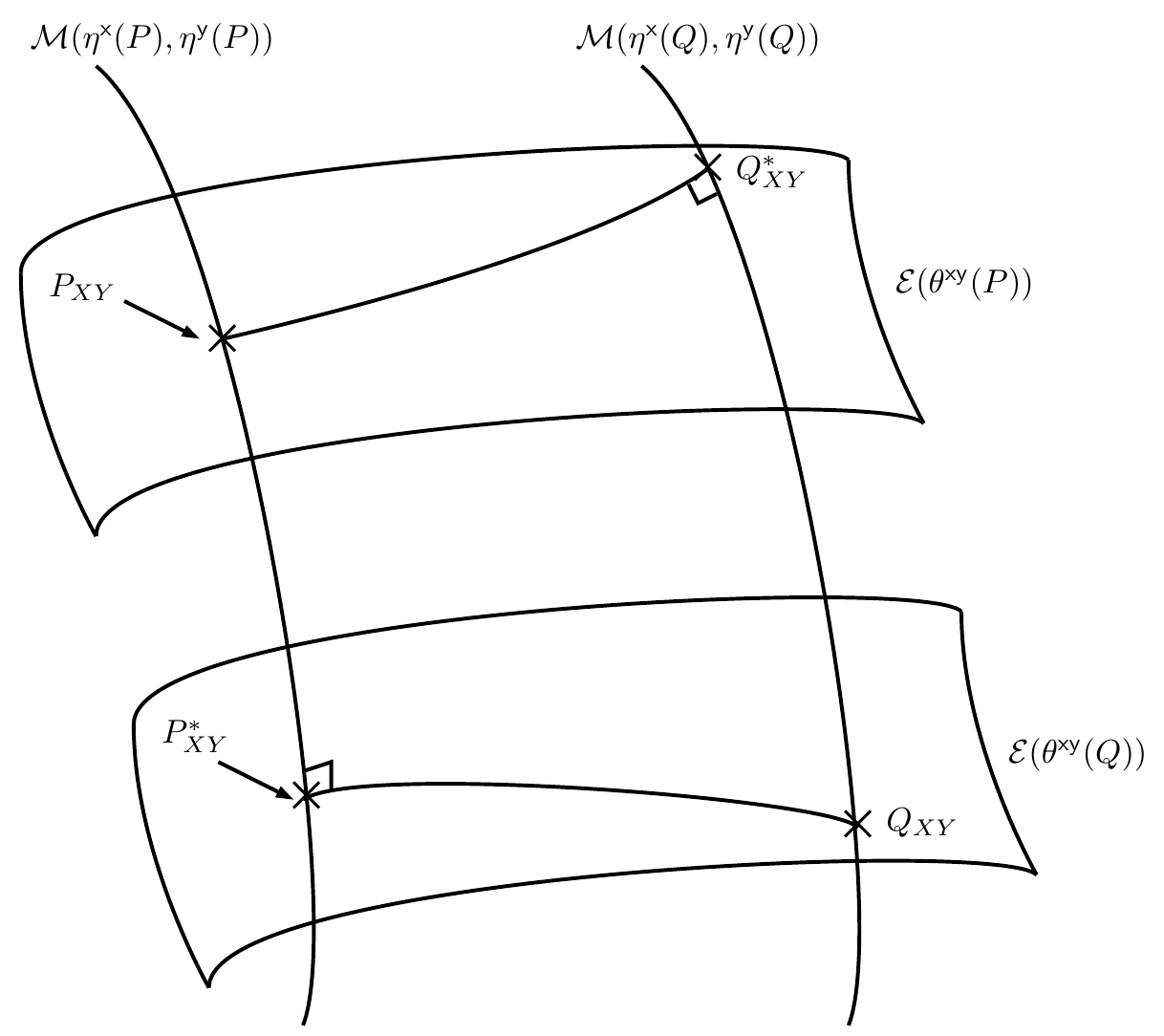}
\caption{A description of Theorem \ref{theorem:pythagorean}. The mixture family ${\cal M}(\eta^\san{x}(P),\eta^\san{y}(P))$ containing $P_{XY}$ and 
the exponential family ${\cal E}(\theta^{\san{xy}}(Q))$ containing $Q_{XY}$ intersect at $P_{XY}^*$. On the other hand,
the mixture family ${\cal M}(\eta^\san{x}(Q),\eta^\san{y}(Q))$ containing $Q_{XY}$ and the exponential family ${\cal E}(\theta^{\san{xy}}(P))$ containing $P_{XY}$
intersect at $Q_{XY}^*$. Those intersections are orthogonal in the sense of \eqref{eq:pythagorean} and \eqref{eq:pythagorean-2}, respectively.}
\label{Fig:pythagorean}
%\end{minipage}
}
\end{figure}
%%%
	
	%%%%%%%%%%%%%%%%%%%%%%%%%%%%%%%%%%%%%%%%%%%
	\section{Neyman-Pearson-Like Testing Scheme} \label{section:non-asymptotic}
	
	In this section, we propose a new testing scheme and evaluate its
	non-asymptotic performance.

	%%%%%%%%%%%%%%
	%\subsection{New Testing Scheme} \label{subsec:new-testing-scheme}

	We first investigate the properties of the exponent function $F(r)$. For that purpose, the following property 
	of $E(P_{XY}\|Q_{XY})$ is useful.
	
	\begin{lemma} \label{lemma:convexity-of-E}
		$E(P_{XY} \| Q_{XY})$ is a convex function with respect to $P_{XY}$. 
	\end{lemma}
	\begin{proof}
		For given $\bar{P}_{XY}$ and $\tilde{P}_{XY}$ and $0 < p < 1$, we have
		\begin{align}
		p E(\bar{P}_{XY} \| Q_{XY}) + (1-p) E(\tilde{P}_{XY} \| Q_{XY}) 
		&= p D(\bar{P}^*_{XY} \| Q_{XY}) + (1-p) D(\tilde{P}^*_{XY} \| Q_{XY}) \\
		&\ge D( p \bar{P}^*_{XY} + (1-p) \tilde{P}^*_{XY} \| Q_{XY}) \\
		&\ge E( p \bar{P}_{XY} + (1-p) \tilde{P}_{XY} \| Q_{XY}),
		\end{align}
		where $\bar{P}^*_{XY}$ and $\tilde{P}^*_{XY}$ are optimizers of $E(\bar{P}_{XY} \| Q_{XY})$ and $E(\tilde{P}_{XY} \| Q_{XY})$, respectively, 
		and the last inequality holds since the marginals of $p \bar{P}^*_{XY} + (1-p) \tilde{P}^*_{XY}$
		and $p \bar{P}_{XY} + (1-p) \tilde{P}_{XY} $ coincide.
	\end{proof}
	
	Lemma \ref{lemma:convexity-of-E} implies the set
	\begin{align} \label{eq:definition-C}
	{\cal C} := \{ \tilde{P}_{XY} : E(\tilde{P}_{XY} \| P_{XY}) \le r \}.
	\end{align}
	is a convex set.
	We can also find that $F(0) = E(P_{XY} \| Q_{XY})$ and $F(E(Q_{XY} \| P_{XY})) = 0$.
	$F(r)$ is also convex. For $0 < r < E(Q_{XY} \| P_{XY})$, we have $F(r) > 0$.
	Thus, $F(r)$ is a monotonically decreasing function for $0 \le r \le E(Q_{XY} \| P_{XY})$, and
	the minimization is attained at the boundary, i.e., $\tilde{P}_{XY}$ satisfying $E(\tilde{P}_{XY} \| P_{XY}) = r$.
	Let 
	\begin{align} \label{eq:definition-lambda-r-conversion}
	\lambda(r) := - r + F(r). 
	\end{align}
	Since $F(r)$ is monotonically decreasing, $\lambda(r)$ is a one-to-one mapping for
	$0 \le r \le E(Q_{XY} \| P_{XY})$, and we find that $\lambda(0) = E(P_{XY}\|Q_{XY})$
	and $\lambda(E(Q_{XY} \| P_{XY})) = - E(Q_{XY} \| P_{XY})$.
	
	The following lemma provide some properties of the optimizers in \eqref{eq:exponent-expression-1} and \eqref{eq:exponent-expression-2}.
	\begin{lemma} \label{lemma:properties-optimizer}
	For $0 < r < E(Q_{XY}\|P_{XY})$, the optimizer $Q_{XY}^\star$ of \eqref{eq:exponent-expression-1} and the optimizer $P_{XY}^\star$
	of \eqref{eq:exponent-expression-2} are unique, and those optimizers satisfy $(Q_{XY}^\star, P_{XY}^\star) \in {\cal E}(\theta^{\san{xy}}(Q)) \times {\cal E}(\theta^\san{xy}(P))$
	and have the same marginals, i.e., $Q_X^\star = P_X^\star$ and $Q_Y^\star = P_Y^\star$.
	\end{lemma}
	\begin{proof}
	For simplicity of notation, we denote the values of \eqref{eq:exponent-expression-1} and \eqref{eq:exponent-expression-2} by $v_1$ and $v_2$, respectively. 
	Since the set ${\cal C}$ defined in \eqref{eq:definition-C} is a convex set, Theorem \ref{theorem:extremal} implies that 
	there exists a unique optimizer $Q^\star_{XY}$ of \eqref{eq:exponent-expression-1}.\footnote{Since ${\cal C}$ contains
	$P_{XY}$, which has full support, Theorem \ref{theorem:extremal} implies that $Q^\star_{XY}$ has full support.} 
	Note that $Q^\star_{XY}$ must be included in 
	${\cal E}(\theta^{\san{xy}}(Q))$; otherwise, the distribution $Q^\dagger_{XY}$ obtained by projecting\footnote{Here, the projection means
	$Q^\dagger_{XY}$ satisfying $Q^\dagger_X = Q^\star_X$, $Q^\dagger_Y = Q^\star_Y$ and $Q^\dagger_{XY} \in {\cal E}(\theta^\san{xy}(Q))$.} 
	$Q^\star_{XY}$
	onto ${\cal E}(\theta^{\san{xy}}(Q))$ satisfies $Q^\dagger_{XY} \in {\cal C}$ and
	\begin{align} \label{eq:condition-of-optimal-projection}
	D(Q^\dagger_{XY} \| Q_{XY}) < D(Q^\star_{XY} \| Q_{XY}),
	\end{align}
	which contradict the fact that $Q^\star_{XY}$ is the optimizer of \eqref{eq:exponent-expression-1}.
	
	Let $P_{XY}^\star \in {\cal E}(\theta^\san{xy}(P))$ be the counterpart of $Q_{XY}^\star$ satisfying $Q_X^\star= P_X^\star$ and $Q_Y^\star = P_Y^\star$.
	Then, since 
	\begin{align}
	D(P_{XY}^\star \| P_{XY}) = E(Q_{XY}^\star \| P_{XY}) = r
	\end{align}
	and $E(P_{XY}^\star \| Q_{XY}) = D(Q_{XY}^\star \| Q_{XY})$, we have $v_1 \ge v_2$.
	
	Let $P_{XY}^\prime$ be an optimizer of \eqref{eq:exponent-expression-2}, 
	and let $Q_{XY}^\prime$ be the distribution such that $D(Q_{XY}^\prime \| Q_{XY}) = E(P_{XY}^\prime \| Q_{XY})$
	and $Q_X^\prime = P_X^\prime$ and $Q_Y^\prime = P_Y^\prime$.\footnote{At this point, it is not guaranteed that $P_{XY}^\prime$
	has full support and that $Q_{XY}^\prime \in {\cal E}(\theta^\san{xy}(Q))$.} Since $E(Q_{XY}^\prime \| P_{XY}) \le D(P_{XY}^\prime \| P_{XY}) \le r$, we have $v_1 \le v_2$.
	
	Since $v_1 = v_2$, $Q_{XY}^\prime$ is also the optimizer of \eqref{eq:exponent-expression-1}.
	However, since the optimizer of \eqref{eq:exponent-expression-1} is unique, we have $Q_{XY}^\prime = Q_{XY}^\star$. 
	We can also verify that $P_{XY}^\prime = P_{XY}^\star$ as follows. Note that $F(r)$ is achieved at the boundary, i.e., 
	$E(Q_{XY}^\star \| P_{XY}) = r$. Since the marginals of $P_{XY}^\prime$, $Q_{XY}^\prime = Q_{XY}^\star$, and $P_{XY}^\star$ are the same, 
	and since $P_{XY}^\star$ is the unique optimizer of $\min\{ D(P_{\bar{X}\bar{Y}} \| P_{XY}) : P_{\bar{X}} = P_X^\star, P_{\bar{Y}} = P_X^\star\}$ (see Theorem \ref{theorem:extremal}), 
	if $P_{XY}^\prime \neq P_{XY}^\star$, then we have
	\begin{align}
	D(P_{XY}^\star \| P_{XY}) < D(P_{XY}^\prime \| P_{XY}),
	\end{align}
	which contradict the fact that $D(P_{XY}^\star \| P_{XY}) = E(Q_{XY}^\star \| P_{XY}) = r$ and $D(P_{XY}^\prime \| P_{XY}) \le r$. 
	Consequently, the optimizer of \eqref{eq:exponent-expression-2} is also unique, and it is included in ${\cal E}(\theta^\san{xy}(P))$.
	Finally, the claim that $Q_X^\star = P_X^\star$ and $Q_Y^\star = P_Y^\star$ is apparent from the above argument. 
	\end{proof}
	
	By using Lemma \ref{lemma:properties-optimizer}, we can show 
	the following geometrical properties of the optimizers in \eqref{eq:exponent-expression-1} and \eqref{eq:exponent-expression-2},
	and it plays an important role in the subsequent construction of our new testing scheme. 
	
	\begin{theorem}
		For $0 < r < E(Q_{XY}\|P_{XY})$, the optimization problem $F(r)$ in \eqref{eq:exponent-expression-1} and \eqref{eq:exponent-expression-2} is
		obtained by the unique pair 
		$(Q^\lambda_{XY}, P^\lambda_{XY}) \in {\cal E}(\theta^{\san{xy}}(Q)) \times {\cal E}(\theta^\san{xy}(P))$ 
		satisfying
		the following equations for some $a \in \mathbb{R}\backslash\{0\}$, $b \in \mathbb{R}$, and 
		and $\lambda = \lambda(r)$:
		\begin{align}
		\log \frac{P^{\lambda}_{XY}(x,y)}{P_{XY}(x,y)} &= a \log \frac{Q^{\lambda}_{XY}(x,y)}{Q_{XY}(x,y)} + b,~~~\forall (x,y) \in {\cal X}\times {\cal Y}, \label{eq:condition-alignment} \\
		D(Q^\lambda_{XY} \| Q_{XY}) - D(P^\lambda_{XY} \| P_{XY}) &= \lambda, \label{eq:condition-lambda} \\
		\sum_{x } P^\lambda_{XY}(x,y) &= \sum_{x } Q^\lambda_{XY}(x,y), \label{eq:condition-pair-1} \\
		\sum_{y } P^\lambda_{XY}(x,y) &= \sum_{y } Q^\lambda_{XY}(x,y), \label{eq:condition-pair-2}
		\end{align}
		and
		\begin{align}
		\sum_{x,y} Q_{XY}(x,y) \Lambda_\lambda(x,y)
		< \lambda < \sum_{x,y} P_{XY}(x,y) \Lambda_\lambda(x,y),
		\label{eq:condition-range}
		\end{align}
		where 
		\begin{align} \label{eq:definition-proxy-llr}
		\Lambda_\lambda(x,y) := \log \frac{Q_{XY}^\lambda(x,y)}{Q_{XY}(x,y)} - \log \frac{P_{XY}^\lambda(x,y)}{P_{XY}(x,y)}.
		\end{align}
	\end{theorem}
	%%%
	\begin{proof}
		First, we show that the optimizer pair of \eqref{eq:exponent-expression-1} and \eqref{eq:exponent-expression-2}
		satisfy \eqref{eq:condition-alignment}-\eqref{eq:condition-range}.
		From Lemma \ref{lemma:properties-optimizer}, the optimizer pair of \eqref{eq:exponent-expression-1} and \eqref{eq:exponent-expression-2}
		are given by unique pair $(Q_{XY}^\lambda, P_{XY}^\lambda) \in {\cal E}(\theta^\san{xy}(Q)) \times {\cal E}(\theta^\san{xy}(P))$ 
		satisfying \eqref{eq:condition-pair-1} and \eqref{eq:condition-pair-2}. 
		Let 
		\begin{align}
		%{\cal S}_=(Q^\lambda_{XY}) &:= \bigg\{ \tilde{Q}_{XY} : \sum_{x,y} (\tilde{Q}_{XY}(x,y) - Q^\lambda_{XY}(x,y)) \log \frac{Q^\lambda_{XY}(x,y)}{Q_{XY}(x,y)} = 0 \bigg\}, \\
		{\cal S}_{\ge}(Q^\lambda_{XY}) &:= \bigg\{ \tilde{Q}_{XY} : \sum_{x,y} (\tilde{Q}_{XY}(x,y) - Q^\lambda_{XY}(x,y)) \log \frac{Q^\lambda_{XY}(x,y)}{Q_{XY}(x,y)} \ge 0 \bigg\};
		\label{eq:supporting-set-of-Q-lambda}
		\end{align}
		we also define ${\cal S}_=(Q^\lambda_{XY})$ and ${\cal S}_\le(Q^\lambda_{XY})$ by replacing $\ge$ with $=$ and $\le$ in \eqref{eq:supporting-set-of-Q-lambda}.
		By Corollary \ref{corollary:implication-pythagorean}, we have $P^{\lambda}_{XY} \in {\cal S}_=(Q^\lambda_{XY})$.
		Since $F(r)$ is achieved at the boundary, i.e., $E(Q_{XY}^\lambda \|P_{XY}) = r$, we have
		\begin{align} 
		D(P^\lambda_{XY} \| P_{XY}) = E(Q_{XY}^\lambda \|P_{XY}) 
		= r. \label{eq:saturating-constraint}
		\end{align}
		Furthermore, we also have
		\begin{align}
		\lefteqn{ \min\big\{ D(\tilde{P}_{XY} \| P_{XY}) : \tilde{P}_{XY} \in {\cal S}_=(Q^\lambda_{XY}) \big\} } \label{eq:optimality-of-tilde-P-star} \\
		&\ge \inf\big\{ D(\tilde{P}_{XY} \| P_{XY}) : \tilde{P}_{XY} \notin {\cal S}_\ge(Q^\lambda_{XY}) \big\}  \\
		&\ge \inf\big\{ D(\tilde{P}_{XY} \| P_{XY}) : \tilde{P}_{XY} \notin {\cal C} \big\} \\
		&\ge \inf\big\{ E(\tilde{P}_{XY} \| P_{XY}) : \tilde{P}_{XY} \notin {\cal C} \big\} \\
		&\ge r, \label{eq:inequality-optimality-of-tilde-P-star}
		\end{align}
		where the second inequality follows from ${\cal C} \subseteq {\cal S}_\ge(Q^\lambda_{XY})$ (see \eqref{eq:inclusion-convex-supporting}).
		The inequalities \eqref{eq:optimality-of-tilde-P-star}-\eqref{eq:inequality-optimality-of-tilde-P-star} together with \eqref{eq:saturating-constraint} 
		imply that $P^\lambda_{XY}$ is the optimizer of \eqref{eq:optimality-of-tilde-P-star}. Thus, from Theorem \ref{theorem:pythagorean-linear-family},
		$P^\lambda_{XY}$ is contained in the exponential family generated by
		$P_{XY}$ and $\log \frac{Q^\lambda_{XY}(x,y)}{Q_{XY}(x,y)}$, i.e., it can be written as 
		\begin{align}
		P^\lambda_{XY}(x,y) = P_{XY}(x,y) \exp\bigg[ s \log \frac{Q^\lambda_{XY}(x,y)}{Q_{XY}(x,y)}  - \psi(s) \bigg]
		\end{align} 
		for some $s,\psi(s) \in \mathbb{R}$, which implies that the pair $(Q^\lambda_{XY},P^\lambda_{XY})$ satisfies \eqref{eq:condition-alignment}.\footnote{Note that
		$s \neq 0$ since $P_{XY}^\lambda \neq P_{XY}$.} 
		The pair also satisfies \eqref{eq:condition-lambda} since $D(Q^\lambda_{XY}\|Q_{XY}) = F(r)$ and \eqref{eq:saturating-constraint} (see also \eqref{eq:definition-lambda-r-conversion}).
		Finally, we can confirm \eqref{eq:condition-range} as follows. By noting $P_{XY} \in {\cal S}_\ge(Q_{XY}^\lambda)$ and $D(P_{XY}\|P_{XY}^\lambda) > - D(P_{XY}^\lambda \| P_{XY})$, we have
		\begin{align}
		\sum_{x,y} P_{XY}(x,y) \Lambda_\lambda(x,y)
		&= \sum_{x,y} P_{XY}(x,y)\bigg[ \log \frac{Q_{XY}^\lambda(x,y)}{Q_{XY}(x,y)} - \log \frac{P_{XY}^\lambda(x,y)}{P_{XY}(x,y)} \bigg] \\
		&> D(Q_{XY}^\lambda \| Q_{XY}) - D(P_{XY}^\lambda \| P_{XY}) \\
		&= \lambda.
		\end{align}
		Let ${\cal S}_\ge(P_{XY}^\lambda)$ be the set defined by replacing $Q_{XY}$ and $Q_{XY}^\lambda$ with $P_{XY}$ and $P_{XY}^\lambda$ in \eqref{eq:supporting-set-of-Q-lambda}.
		In fact, by noting \eqref{eq:condition-alignment} and
		\begin{align}
		\sum_{x,y} Q^\lambda_{XY}(x,y) \log \frac{Q^\lambda_{XY}(x,y)}{Q_{XY}(x,y)} = \sum_{x,y} P^\lambda_{XY}(x,y) \log \frac{Q^\lambda_{XY}(x,y)}{Q_{XY}(x,y)},
		\end{align}
		which follows from Corollary \ref{corollary:implication-pythagorean}, we can find that 
		either ${\cal S}_\ge(P_{XY}^\lambda) = {\cal S}_\ge(Q_{XY}^\lambda)$ or ${\cal S}_\ge(P_{XY}^\lambda) = {\cal S}_\le(Q_{XY}^\lambda)$;
		since $P_{XY} \notin {\cal S}_\ge(P_{XY}^\lambda)$, we have ${\cal S}_\ge(P_{XY}^\lambda) = {\cal S}_\le(Q_{XY}^\lambda)$.
		Thus, by noting $Q_{XY} \in {\cal S}_\ge(P_{XY}^\lambda)$ and $- D(Q_{XY} \| Q_{XY}^\lambda) < D(Q_{XY}^\lambda \| Q_{XY})$, we have
		\begin{align}
		\sum_{x,y} Q_{XY}(x,y) \Lambda_\lambda(x,y)
		&= \sum_{x,y} Q_{XY}(x,y) \bigg[ \log \frac{Q_{XY}^\lambda(x,y)}{Q_{XY}(x,y)} - \log \frac{P_{XY}^\lambda(x,y)}{P_{XY}(x,y)} \bigg] \\
		&< D(Q_{XY}^\lambda \| Q_{XY}) - D(P_{XY}^\lambda \| P_{XY}) \\
		&= \lambda.
		\end{align}
		
		Second, for a given pair $(Q^\lambda_{XY}, P^\lambda_{XY})$ satisfying \eqref{eq:condition-alignment}-\eqref{eq:condition-range},
		we show that the pair is the optimizer pair of \eqref{eq:exponent-expression-1} and \eqref{eq:exponent-expression-2};
		then, since the optimizer pair of \eqref{eq:exponent-expression-1} and \eqref{eq:exponent-expression-2} is unique,
		the solution pair of \eqref{eq:condition-alignment}-\eqref{eq:condition-range} is also unique.
		
		We shall show that $F(D(P^\lambda_{XY} \| P_{XY})) = D(Q^\lambda_{XY} \| Q_{XY})$. Then, since
		$\lambda(r)$ is one-to-one, $E(Q_{XY}^\lambda \| P_{XY}) = D(P^\lambda_{XY} \| P_{XY}) = r$ and 
		$E(P_{XY}^\lambda \| Q_{XY}) = D(Q^\lambda_{XY} \| Q_{XY}) =  F(r)$, which implies 
		that the pair is the optimizer pair of \eqref{eq:exponent-expression-1} and \eqref{eq:exponent-expression-2}.
		
		Since $Q_{XY}^\lambda$ is contained in the set ${\cal C}$ of \eqref{eq:definition-C}
		for $r = D(P_{XY}^\lambda \| P_{XY})$, we have $D(Q_{XY}^\lambda \| Q_{XY}) \ge F(D(P_{XY}^\lambda \| P_{XY}))$.
		In order to prove the opposite inequality, let $Q_{XY}^\star$ be the optimizer of \eqref{eq:exponent-expression-1} for $r=D(P_{XY}^\lambda \| P_{XY})$.
		From Lemma \ref{lemma:properties-optimizer}, we have $Q_{XY}^\star \in {\cal E}(\theta^\san{xy}(Q))$.
		Let $P_{XY}^\star \in {\cal E}(\theta^\san{xy}(P))$ be the same marginal counterpart of $Q_{XY}^\star$, i.e., the marginals satisfy $P_X^\star = Q_X^\star$
		and $P_Y^\star = Q_Y^\star$. Note that
		\begin{align}
		D(P_{XY}^\star \| P_{XY}) = E(Q_{XY}^\star \| P_{XY}) \le D(P_{XY}^\lambda \| P_{XY})
		\end{align}
		since $Q_{XY}^\star$ is contained in the set ${\cal C}$ of \eqref{eq:definition-C} for $r = D(P_{XY}^\lambda \| P_{XY})$. Thus, we have
		\begin{align}
		D(P_{XY}^\star \| P_{XY}) - D(P_{XY}^\star \| P_{XY}^\lambda) \le D(P_{XY}^\lambda \| P_{XY}).
		\end{align}
		Furthermore, from Corollary \ref{corollary:implication-pythagorean} and the fact that $P_X^\star = Q_X^\star$
		and $P_Y^\star = Q_Y^\star$, we have
		\begin{align}
		\lefteqn{ D(P_{XY}^\star \| P_{XY}) - D(P_{XY}^\star \| P_{XY}^\lambda) } \\
		&= \sum_{x,y} Q_{XY}^\star(x,y) \bigg[ \log \frac{P_{XY}^\star(x,y)}{P_{XY}(x,y)} - \log \frac{P_{XY}^\star(x,y)}{P_{XY}^\lambda(x,y)} \bigg] \\
		&= \sum_{x,y} Q_{XY}^\star(x,y) \log \frac{P_{XY}^\lambda(x,y)}{P_{XY}(x,y)}.
		\end{align}
		Thus, we have $Q_{XY}^\star \in {\cal S}_{\le}(P_{XY}^\lambda)$, where ${\cal S}_{\le}(P_{XY}^\lambda)$ is defined by replacing $\ge$ with $\le$
		in the definition of ${\cal S}_{\ge}(P_{XY}^\lambda)$. As we have shown above, \eqref{eq:condition-alignment} and Corollary \ref{corollary:implication-pythagorean}
		imply ${\cal S}_{\le}(P_{XY}^\lambda) = {\cal S}_{\ge}(Q_{XY}^\lambda)$. Thus, we have
		\begin{align}
		F(D(P_{XY}^\lambda \| P_{XY})) &= D(Q_{XY}^\star \| Q_{XY}) \\
		&\ge \min\big\{ D(Q_{\bar{X}\bar{Y}} \| Q_{XY}) : Q_{\bar{X}\bar{Y}} \in {\cal S}_{\ge}(Q_{XY}^\lambda)  \big\} \\
		&= D(Q_{XY}^\lambda \| Q_{XY}),
		\end{align}
		where the last equality follows since $Q_{\bar{X}\bar{Y}} \in {\cal S}_{\ge}(Q_{XY}^\lambda)$ is equivalent to
		\begin{align}
		D(Q_{\bar{X}\bar{Y}} \| Q_{XY}) \ge D(Q_{\bar{X}\bar{Y}} \| Q_{XY}^\lambda) + D(Q_{XY}^\lambda \| Q_{XY}).
		\end{align}
	\end{proof}

	Now, we introduce a proxy of the log-likelihood ratio $\Lambda_\lambda(x,y)$ as follows. 
	For $- E(Q_{XY} \| P_{XY}) < \lambda < E(P_{XY} \| Q_{XY})$, we define 
	\begin{align}
	\Lambda_\lambda(x,y) := \log \frac{Q_{XY}^\lambda(x,y)}{Q_{XY}(x,y)} - \log \frac{P_{XY}^\lambda(x,y)}{P_{XY}(x,y)}
	\end{align}
	for the unique solution pair $(Q_{XY}^\lambda,P_{XY}^\lambda)$
	satisfying \eqref{eq:condition-alignment}-\eqref{eq:condition-range}; for $\lambda = E(P_{XY} \| Q_{XY})$, we define
	\begin{align}
	\Lambda_\lambda(x,y) := \log \frac{P_{XY}^*(x,y)}{Q_{XY}(x,y)},
	\end{align}
	where $P_{XY}^*$ is the optimizer of $E(P_{XY}\|Q_{XY})$; for $\lambda = - E(Q_{XY} \| P_{XY})$, we define 
	\begin{align}
	\Lambda_\lambda(x,y) := - \log \frac{Q^*_{XY}(x,y)}{P_{XY}(x,y)},
	\end{align}
	where $Q_{XY}^*$ is the optimizer of $E(Q_{XY} \| P_{XY})$.

	Since $(Q^\lambda_{XY}, P^\lambda_{XY}) \in {\cal E}(\theta^{\san{xy}}(Q)) \times {\cal E}(\theta^\san{xy}(P))$,\footnote{Note also that 
		$P^*_{XY} \in {\cal E}(\theta^{\san{xy}}(Q))$ and $Q^*_{XY} \in {\cal E}(\theta^\san{xy}(P))$.} i.e., the pairs 
	$Q^\lambda_{XY}$ and $P^\lambda_{XY}$ have the same values as $Q_{XY}$ and $P_{XY}$ at the
	$\theta^{\san{x}\san{y}}$-coordinate, respectively, we can write 
	\begin{align} \label{eq:decomposition-llr}
	\Lambda_\lambda(x,y) = a_1(x) + a_2(y)
	\end{align} 
	for some functions $a_1$ on ${\cal X}$ and $a_2$ on ${\cal Y}$. 
	Thus, for any joint distribution $P_{\bar{X}\bar{Y}} \in \overline{{\cal P}}({\cal X}\times {\cal Y})$ with marginals $P_{\bar{X}}$ and $P_{\bar{Y}}$, it holds that
	\begin{align}
	\sum_{x,y} P_{\bar{X}}(x) P_{\bar{Y}}(y) \Lambda_\lambda(x,y)
	&= \sum_x P_{\bar{X}}(x) a_1(x) + \sum_y P_{\bar{Y}}(y) a_2(y) \\
	&= \sum_{x,y} P_{\bar{X}\bar{Y}}(x,y) \Lambda_\lambda(x,y). \label{eq:preserve-expectation}
	\end{align}
	More generally, $\sum_{x,y} \tilde{P}_{\bar{X}\bar{Y}}(x,y) \Lambda_\lambda(x,y)$
	and $\sum_{x,y} P_{\bar{X}\bar{Y}}(x,y) \Lambda_\lambda(x,y)$ take the same value as long as the marginals
	of $\tilde{P}_{\bar{X}\bar{Y}}$ and $P_{\bar{X}\bar{Y}}$ coincide. 
	Geometrically, this is because ${\cal M}(\eta^\san{x},\eta^\san{y})$ for a given $(\eta^\san{x},\eta^\san{y})$ is orthogonal to
	${\cal E}(\theta^\san{xy}(P))$ and ${\cal E}(\theta^{\san{xy}}(Q))$.

	\begin{remark} \label{remark:crucial-identity}
		Functions $a_1,a_2$ in 
		decomposition \eqref{eq:decomposition-llr} are not unique; for instance, by some calculation, we can verify (see Appendix \ref{appendix:proof-of-remark:crucial-identity}) that
		$\Lambda_\lambda(x,y)$ can be decomposed as\footnote{Note that 
			${\cal X} = \{0,1\,\ldots, |{\cal X}|-1\}$ and ${\cal Y} = \{ 0,1,\ldots,|{\cal Y}|-1\}$. In fact, the choice $(0,0)$ is not crucial,
			and the same statement holds even if we replace $\Lambda_\lambda(x,0)$, $\Lambda_\lambda(0,y)$, and
			$\Lambda_\lambda(0,0)$ with $\Lambda_\lambda(x,j)$, $\Lambda_\lambda(i,y)$, and $\Lambda_\lambda(i,j)$ for arbitrarily fixed $(i,j) \in {\cal X} \times {\cal Y}$, respectively.}   
		\begin{align}
		\Lambda_\lambda(x,y) = \Lambda(x,0) + \Lambda_\lambda(0,y) - \Lambda_\lambda(0,0)
		\end{align}
	\end{remark}

	Now, we are ready to propose our testing scheme. Fix arbitrary $- E(Q_{XY} \| P_{XY}) \le \lambda \le E(P_{XY} \| Q_{XY})$.
	Upon observing $\bm{x}$ and $\bm{y}$, the encoders send their types, i.e., 
	$f_1(\bm{x}) = \san{t}_{\bm{x}}$ and $f_2(\bm{y}) = \san{t}_{\bm{y}}$, respectively. Then, upon receiving a pair of marginal types $(\san{t}_{\bm{x}},\san{t}_{\bm{y}})$,
	the decoder $g$ outputs $\san{H}_0$ if
	\begin{align}
	\sum_{x,y} \san{t}_{\bm{x}}(x)\san{t}_{\bm{y}}(y) \Lambda_{\lambda}(x,y)  > \tau
	\end{align}
	holds for a prescribed threshold $\tau$; otherwise it outputs $\san{H}_1$.
	By \eqref{eq:preserve-expectation}, the decoder outputs $g(f_1(\bm{x}),f_2(\bm{y})) = \san{H}_0$ if and only if
	\begin{align}
	\sum_{x,y} \san{t}_{\bm{x}\bm{y}}(x,y) \Lambda_{\lambda}(x,y) 
	&= \frac{1}{n} \sum_{i=1}^n \Lambda_{\lambda}(x_i,y_i) \\
	&>\tau.
	\end{align}
	Apparently, this scheme is a symmetric scheme. 
	The performance of this scheme is summarized in the following theorem.

	\begin{theorem} \label{theorem:performance-NP-type-test}
		For a given $- E(Q_{XY} \| P_{XY}) \le \lambda \le E(P_{XY} \| Q_{XY})$ and $\tau \in \mathbb{R}$, the above mentioned Neyman-Pearson-like scheme $T_n^\mathtt{NPl}$ 
		has the following error trade-off: 
		\begin{align}
		\alpha[T_n^\mathtt{NPl}] &= P\bigg( \frac{1}{n} \sum_{i=1}^n \Lambda_{\lambda}(X_i,Y_i) \le \tau \bigg), \\
		\beta[T_n^\mathtt{NPl}] &= Q\bigg( \frac{1}{n} \sum_{i=1}^n \Lambda_{\lambda}(X_i,Y_i) > \tau \bigg).
		\end{align}
	\end{theorem}
	
	Although it is not clear whether the Neyman-Pearson-like test is optimal or not for a given blocklength,
	we will numerically examine that the Neyman-Pearson-like test has better error trade-off than the Hoeffding-like test
	in Section \ref{sec:example}. Furthermore, for the large-deviation regime and the second-order regime, 
	we will show that the Neyman-Pearson-like test is asymptotically optimal in Section \ref{sec:LDP} and Section \ref{sec:second-order}, respectively.

	%%%%%

	%%%%%%%%%%%%%%%%%%%%%
	%\subsection{Geometrical Comparison of Two Schemes} \label{subsection:geometrical-comparison}

\begin{figure}[!t]
  \centering
  \subfloat[][]{\includegraphics[width=.4\textwidth, bb=0 0 356 317]{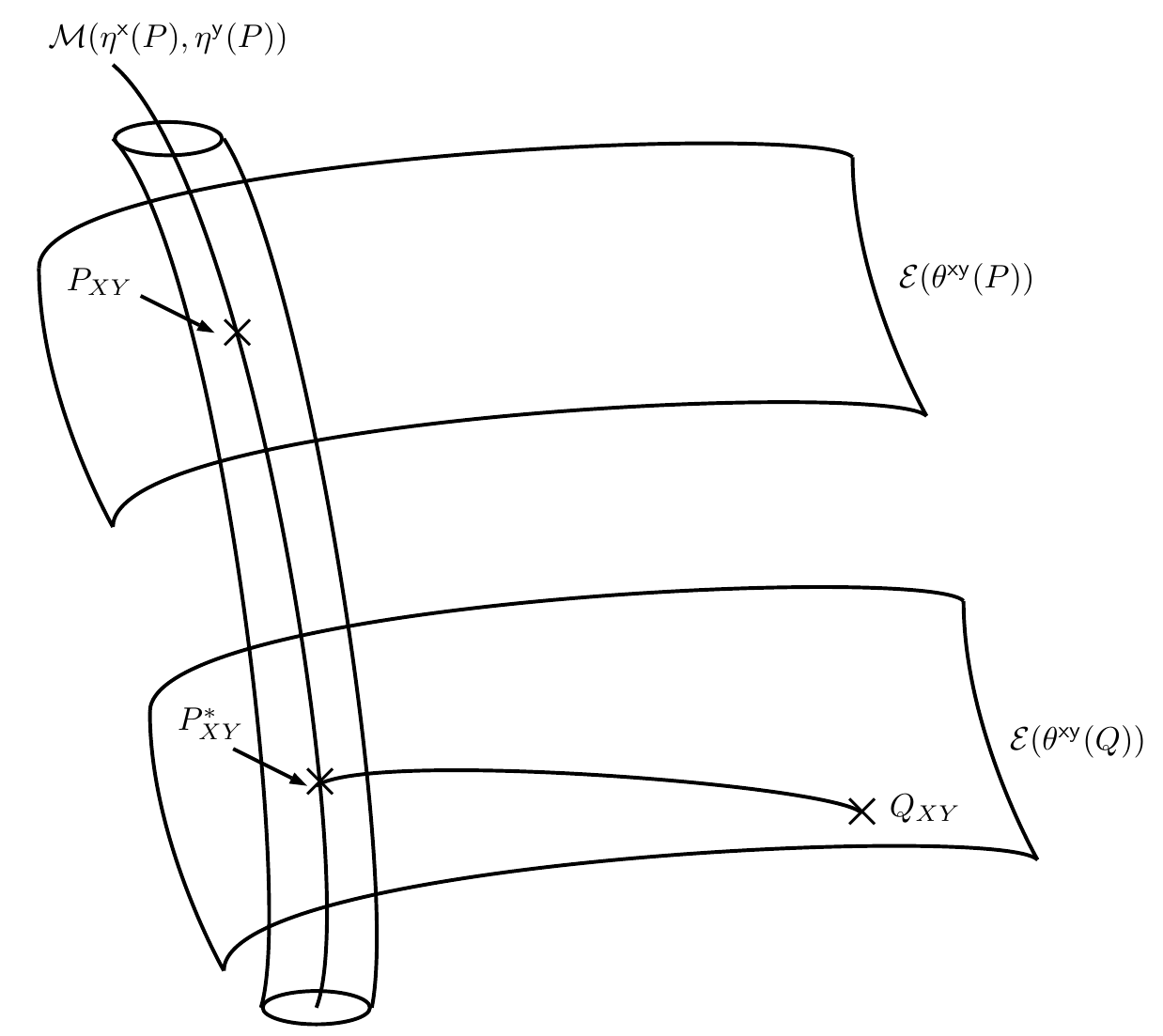} \label{Fig:test}}  \quad
  \subfloat[][]{\includegraphics[width=.4\textwidth, bb=0 0 350 312]{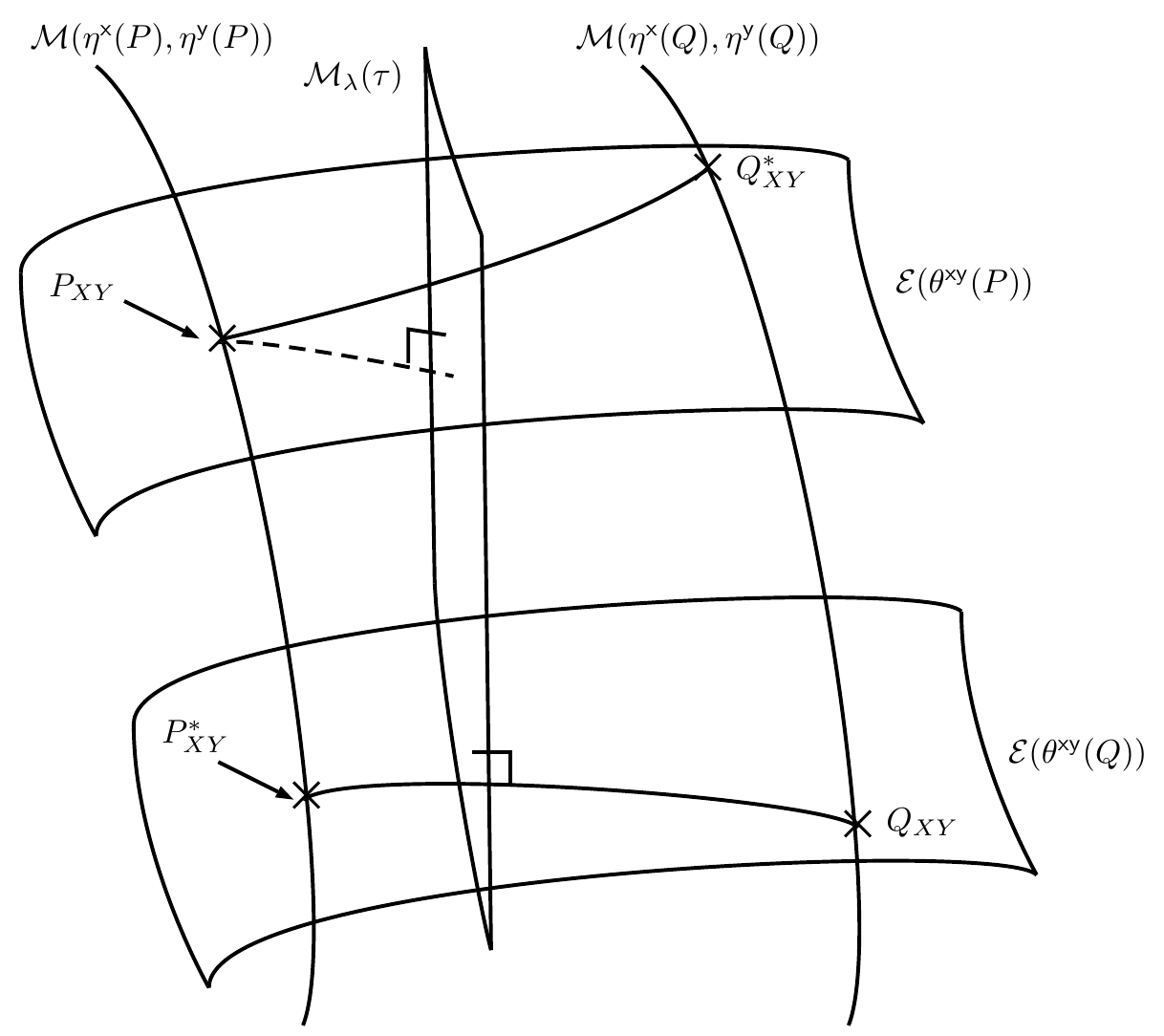} \label{Fig:test3}}  \\
  \subfloat[][]{\includegraphics[width=.4\textwidth, bb=0 0 350 312]{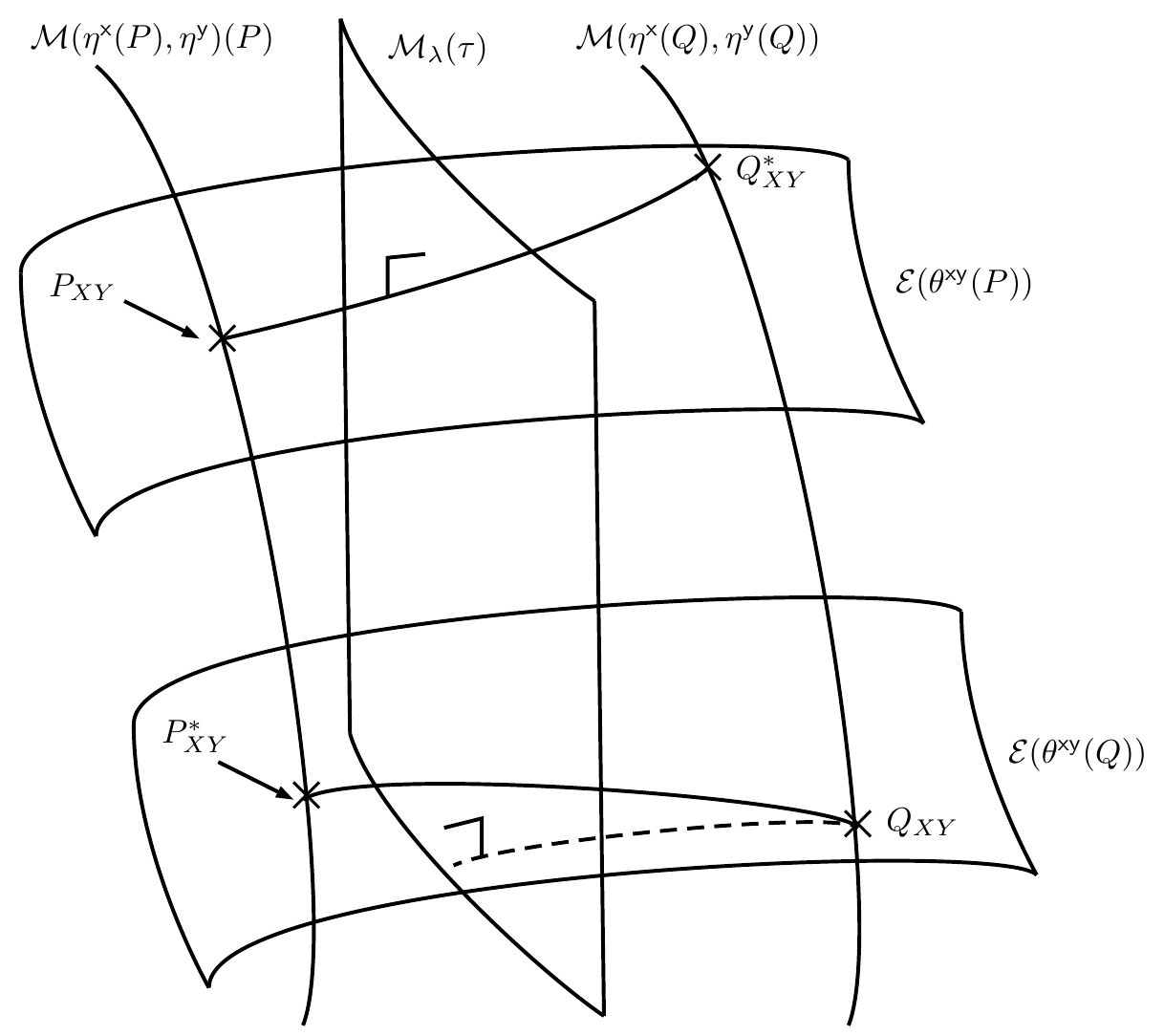} \label{Fig:test2}}  \quad
  \subfloat[][]{\includegraphics[width=.4\textwidth, bb=0 0 350 312]{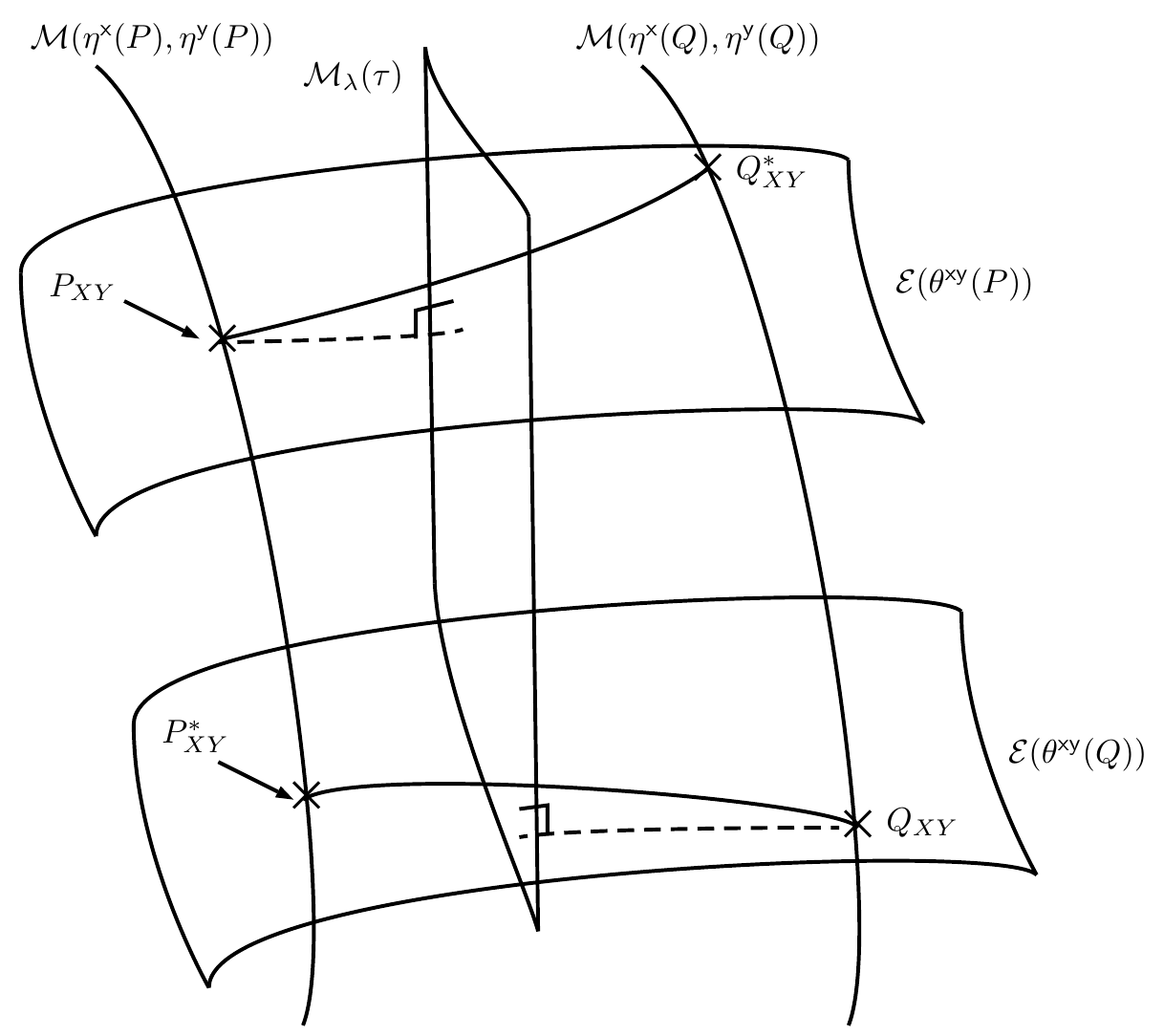} \label{Fig:test4}} 
    \caption{(a) A description of Han-Kobayashi's testing scheme; and descriptions of our proposed scheme with (b) $\lambda= E(P_{XY}\|Q_{XY})$, 
  (c) $\lambda = - E(Q_{XY} \| P_{XY})$, and (d) $- E(Q_{XY} \| P_{XY}) < \lambda < E(P_{XY} \| Q_{XY})$.}
  \label{Fig:testing-geometry}
\end{figure}	
	
	We close this section by comparing Han-Kobayashi's scheme (see Section \ref{subsec:han-kobayashi-scheme}) and the Neyman-Pearson-like scheme proposed above. 
	In Han-Kobayashi's scheme, a type $P_{\bar{X}\bar{Y}}$ is accepted if and only if there exists some $\tilde{P}_{\bar{X}\bar{Y}}$ satisfying 
	$P_{\bar{X}} = \tilde{P}_{\bar{X}}$, $P_{\bar{Y}} = \tilde{P}_{\bar{Y}}$, and $D(\tilde{P}_{\bar{X}\bar{Y}} \| P_{XY}) < r$. In other words, 
	a type is accepted if and only if it is included in the cylinder of radius $r$ given by
	\begin{align} \label{eq:cylinder-formula}
	\bigcup_{\eta: D(P_{XY,\eta} \| P_{XY}) < r} {\cal M}(\eta^\san{x},\eta^\san{y})
	\end{align}
	(see \cite{han-amari:98}); see Fig.~\ref{Fig:test}.
	
	For a given $- E(Q_{XY} \| P_{XY}) \le \lambda \le E(P_{XY} \| Q_{XY})$, let 
	\begin{align}
	Q_{XY,s}^\lambda(x,y) &:= Q_{XY}(x,y) \exp\{ s \Lambda_\lambda(x,y) - \psi_{Q,\lambda}(s) \}, \label{eq:definition-exponential-family-Q-s} \\
	\psi_{Q,\lambda}(s) &:= \log \sum_{x,y} Q_{XY}(x,y) \exp\{ s \Lambda_\lambda(x,y) \},
	\end{align}
	and let 
	\begin{align}
	{\cal E}_\lambda(Q) := \big\{ Q_{XY,s}^\lambda : s \in \mathbb{R} \big\}
	\end{align}
	be the exponential family containing $Q_{XY}$.
	Let 
	\begin{align}
	P_{XY,t}^\lambda(x,y) &:= P_{XY}(x,y) \exp\{ t \Lambda_\lambda(x,y) - \psi_{P,\lambda}(t) \}, \\
	\psi_{P,\lambda}(t) &:= \log \sum_{x,y} P_{XY}(x,y) \exp\{ t \Lambda_\lambda(x,y)  \}, 
	\end{align}
	and let
	\begin{align}
	{\cal E}_\lambda(P) := \big\{ P_{XY,t}^\lambda : t \in \mathbb{R} \big\}
	\end{align}
	be the exponential family containing $P_{XY}$. 
	For $\lambda = \overline{\lambda} := E(P_{XY}\|Q_{XY})$, note that ${\cal E}_{\overline{\lambda}}(Q)$ is the e-geodesic connecting 
	$Q_{XY}$ and $P_{XY}^*$; for $\lambda = \underline{\lambda}:= -E(Q_{XY}\|P_{XY})$, note that ${\cal E}_{\underline{\lambda}}(P)$ is the e-geodesic connecting $P_{XY}$ and $Q_{XY}^*$.
	
	For $\tau \in \mathbb{R}$, let 
	\begin{align}
	{\cal M}_\lambda(\tau) := \bigg\{ P_{\bar{X}\bar{Y}} : \sum_{x,y} P_{\bar{X}\bar{Y}}(x,y) \Lambda_\lambda(x,y) = \tau \bigg\}
	\end{align}
	be the mixture family generated by $\Lambda_\lambda(x,y)$. 
	%Note that ${\cal M}_\lambda(\tau)$ is orthogonal to ${\cal E}_\lambda(P)$ and ${\cal E}_\lambda(Q)$.
	
	In our Neyman-Pearson-like testing scheme with $\lambda = \overline{\lambda} = E(P_{XY}\|Q_{XY})$, we bisect the entire space with 
	${\cal M}_\lambda(\tau)$ that is orthogonal\footnote{Here, ``orthogonal" means that the Pythagorean theorem with respect to the relative entropy
		holds at the intersection.} 
	to ${\cal E}_{\overline{\lambda}}(Q)$, the e-geodesic connecting $Q_{XY}$ and $P_{XY}^*$; see Fig.~\ref{Fig:test3}.
	On the other hand, in our Neyman-Pearson testing scheme with $\lambda =-E(Q_{XY}\|P_{XY})$, we bisect the entire space with ${\cal M}_\lambda(\tau)$
	that is orthogonal to ${\cal E}_{\underline{\lambda}}(P)$, the e-geodesic connecting $P_{XY}$ and $Q_{XY}^*$;
	see Fig.~\ref{Fig:test2}. Furthermore, for $-E(Q_{XY} \| P_{XY}) < \lambda < E(P_{XY}\|Q_{XY})$, 
	we bisect the entire space with ${\cal M}_\lambda(\tau)$ that is orthogonal to neither ${\cal E}_{\overline{\lambda}}(Q)$
	nor ${\cal E}_{\underline{\lambda}}(P)$; see Fig.~\ref{Fig:test4}. In fact, the mixture family ${\cal M}_\lambda(\tau)$ is orthogonal to
	${\cal E}_\lambda(Q)$ and ${\cal E}_\lambda(P)$. In particular, when $\tau = \lambda$, 
	the intersections are given by $Q^\lambda_{XY}$ and $P^\lambda_{XY}$, respectively.
	
	In contrast to the standard hypothesis testing problem, 
	%since the e-geodesic connecting $Q_{XY}$ and $P_{XY}^*$ and
	%the e-geodesic connecting $P_{XY}$ and $Q_{XY}^*$ are not aligned, there is a freedom to choose the direction of bisection. 
	our Neyman-Pearson-like testing scheme has the freedom to choose the direction of bisection with parameter $\lambda$. In fact, 
	as we determine in later sections, an appropriate choice of $\lambda$ depending on a target threshold $\tau$
	is very important. 
	
	\begin{remark} \label{remark:parallel-condition}
	By using the coordinate notation, we can rewrite (see Appendix \ref{appendix:proof-remark:parallel-condition}) the condition in \eqref{eq:condition-alignment} as
	\begin{align}
	\big[ \theta^{\san{x}}_i(P_{XY}^\lambda) - \theta^{\san{x}}_i(P_{XY}) \big] &= a\big[ \theta^\san{x}_i(Q_{XY}^\lambda) - \theta^\san{x}_i(Q_{XY}) \big],~\forall i=1,\ldots,d_\san{x}, 
	\label{eq:parallel-condition-1} \\
	\big[ \theta^{\san{y}}_j(P_{XY}^\lambda) - \theta^{\san{y}}_j(P_{XY}) \big] &= a\big[ \theta^\san{y}_j(Q_{XY}^\lambda) - \theta^\san{y}_j(Q_{XY}) \big],~\forall j=1,\ldots,d_\san{y}.
	\label{eq:parallel-condition-2}
	\end{align} 
	Thus, we can find that the e-geodesic connecting $P_{XY}$ and $P_{XY}^\lambda$ and that connecting $Q_{XY}$ and $Q_{XY}^\lambda$
	are parallel each other (see also Fig.~\ref{Fig:trajectory} in Section \ref{sec:example}). 
	\end{remark}
	
	\begin{remark}
		When $\theta^{\san{xy}}(P) = \theta^{\san{xy}}(Q)$, we have $P_{XY}^\lambda = Q_{XY}^\lambda$, which implies 
		\begin{align}
		\Lambda_\lambda(x,y) = \log \frac{P_{XY}(x,y)}{Q_{XY}(x,y)}.
		\end{align}
		Thus, our Neyman-Pearson-like testing scheme reduces to the Neyman-Pearson testing scheme of the standard
		hypothesis testing between $P_{XY}$ and $Q_{XY}$. In other words, when the correlation components of $P_{XY}$ and $Q_{XY}$
		are the same and only the marginals are different, then a test based on  zero-rate encoding is as good as a test which is based on
		full-rate encoding. 
	\end{remark}
	
	\begin{remark}
	By noting 
	\begin{align}
	\sum_{x,y} P_{XY}^\lambda(x,y) \Lambda_\lambda(x,y) = \sum_{x,y} Q_{XY}^\lambda(x,y) \Lambda_\lambda(x,y) = \lambda
	\end{align}
	and \eqref{eq:condition-alignment}, we can rewrite the condition
	\begin{align}
	\sum_{x,y} P_{\bar{X}\bar{Y}}(x,y) \Lambda_\lambda(x,y) = \lambda
	\end{align}
	as
	\begin{align}
	\sum_{x,y} P_{\bar{X}\bar{Y}}(x,y) \log \frac{Q_{XY}^\lambda(x,y)}{Q_{XY}(x,y)} = D(Q_{XY}^\lambda \| Q_{XY})
	\end{align}
	or
	\begin{align}
	\sum_{x,y} P_{\bar{X}\bar{Y}}(x,y) \log \frac{P_{XY}^\lambda(x,y)}{P_{XY}(x,y)} = D(P_{XY}^\lambda \| P_{XY}).
	\end{align}
	Thus, the mixture family plane ${\cal M}_\lambda(\lambda)$ is the tangent plane of the cylinder given by \eqref{eq:cylinder-formula}
	(see also \eqref{eq:supporting-set-of-Q-lambda} and \eqref{eq:tangent-plane-single-terminal}).
	\end{remark}
	
	\begin{remark} \label{remark:trivial-test}
	There is a trivial testing scheme that is most powerful among the class of symmetric schemes.
	Upon observing marginal types $(P_{\bar{X}}, P_{\bar{Y}})$, the decoder accept the null hypothesis if
	\begin{align} \label{eq:LLR-trivial}
	\log \frac{P\big( \san{t}_{X^n} = P_{\bar{X}},~\san{t}_{Y^n} = P_{\bar{Y}} \big)}{Q\big( \san{t}_{X^n} = P_{\bar{X}},~\san{t}_{Y^n} = P_{\bar{Y}} \big)} > \tau
	\end{align}
	for a prescribed threshold $\tau$. In fact, this is the Neyman-Pearson test such that the pair of marginal types $(\san{t}_{X^n}, \san{t}_{Y^n})$ is regarded as the observation.
	However, since computing the log-likelihood ratio in \eqref{eq:LLR-trivial} is intractable as the blocklength becomes larger, it is difficult to
	implement this scheme. On the other hand, in the Neyman-Pearson-like scheme introduced above, we only need to compute 
	the empirical average of $\Lambda_\lambda(x,y)$ with respect to the product of marginal types, and it is easier to implement.  
	\end{remark}
	
	%%%%%%%%%%%%%%%%%%%%%%%%%%%%%%%%%%%%%%%%%%%
	\section{Example} \label{sec:example}
	
	In this section, we consider the binary example, i.e., ${\cal X} = {\cal Y} = \{0,1\}$, and compare 
	the error trade-off of our proposed testing scheme and Han-Kobayashi's scheme.
	In the binary case, ${\cal P}({\cal X}\times {\cal Y})$ is parametrized by three parameters $(\theta^\san{x},\theta^\san{y},\theta^\san{xy})$
	or $(\eta^\san{x},\eta^\san{y},\eta^\san{xy})$. 
	For simplicity of notation, we denote $\tilde{\theta} = \theta(P)$ and $\bar{\theta} = \theta(Q)$; similar notations are used for the expectation parameters.
	For given $(\tilde{\eta}^\san{x},\tilde{\eta}^\san{y})$ and $\bar{\theta}^\san{xy}$, 
	by noting \eqref{eq:correlation-coordinate}, the intersection of
	${\cal E}(\bar{\theta}^\san{xy})$ and ${\cal M}(\tilde{\eta}^\san{x},\tilde{\eta}^\san{y})$ can be derived by solving the following equation
	with respect to $\eta^\san{xy}$:
	\begin{align}
	\log \frac{\eta^\san{xy} (1 - \tilde{\eta}^\san{x} - \tilde{\eta}^\san{y} + \eta^\san{xy})}{(\tilde{\eta}^\san{x} - \eta^\san{xy}) (\tilde{\eta}^\san{y} - \eta^\san{xy})} = \bar{\theta}^\san{xy}.
	\end{align}
	The above equation is equivalent to
	\begin{align}
	(e^{\bar{\theta}^\san{xy}} - 1) (\eta^\san{xy})^2 - [  (\tilde{\eta}^\san{x} + \tilde{\eta}^\san{y}) (e^{\bar{\theta}^\san{xy}} -1) + 1 ] \eta^\san{xy} 
	+ e^{\bar{\theta}^\san{xy}} \tilde{\eta}^\san{x} \tilde{\eta}^\san{y}.
	\end{align}
	When $\bar{\theta}^\san{xy} = 0$, i.e., there is no correlation, then $\eta^\san{xy} = \tilde{\eta}^\san{x} \tilde{\eta}^\san{y}$ is the only solution.
	When $\bar{\theta}^\san{xy} \neq 0$, we have two solutions:
	\begin{align}
	\eta^\san{xy} = \frac{[(\tilde{\eta}^\san{x} + \tilde{\eta}^\san{y})(e^{\bar{\theta}^\san{xy}} -1) + 1] 
		\pm \sqrt{[(\tilde{\eta}^\san{x} + \tilde{\eta}^\san{y})(e^{\bar{\theta}^\san{xy}} -1) + 1]^2 - 4 (e^{\bar{\theta}^\san{xy}} -1) e^{\bar{\theta}^\san{xy}} \tilde{\eta}^\san{x} \tilde{\eta}^\san{y} }}{
		2(e^{\bar{\theta}^\san{xy}} -1)}.
	\end{align}
	Here, note that $\eta^\san{xy}$ must satisfy
	\begin{align}
	\eta^\san{xy} &\le \min[ \tilde{\eta}^\san{x},\tilde{\eta}^\san{y} ] \\
	&\le \frac{\tilde{\eta}^\san{x} + \tilde{\eta}^\san{y}}{2}. \label{eq:condition-on-eta-xy}
	\end{align}
	Thus, when $\bar{\theta}^\san{xy} > 0$, then 
	\begin{align}
	\eta^\san{xy} = \frac{[(\tilde{\eta}^\san{x} + \tilde{\eta}^\san{y})(e^{\bar{\theta}^\san{xy}} -1) + 1] 
		- \sqrt{[(\tilde{\eta}^\san{x} + \tilde{\eta}^\san{y})(e^{\bar{\theta}^\san{xy}} -1) + 1]^2 - 4 (e^{\bar{\theta}^\san{xy}} -1) e^{\bar{\theta}^\san{xy}} \tilde{\eta}^\san{x} \tilde{\eta}^\san{y} }}{
		2(e^{\bar{\theta}^\san{xy}} -1)}
	\label{eq:solution-binary-example}
	\end{align}
	is the only valid solution since the other solution violates \eqref{eq:condition-on-eta-xy}.
	On the other hand, note that $\eta^\san{xy}$ is required to be nonnegative. Thus, when $\bar{\theta}^\san{xy} < 0$, then 
	\eqref{eq:solution-binary-example} is the only valid solution in this case too since the other solution is negative.
	
	The use of \eqref{eq:solution-binary-example} enables us to numerically solve \eqref{eq:condition-alignment}-\eqref{eq:condition-range} to 
	find $(Q_{XY}^\lambda,P_{XY}^\lambda)$ for each $\lambda$.
	More specifically, for given parameters $(\eta^\san{x},\eta^\san{y})$, by using \eqref{eq:solution-binary-example}, we can compute
	$(\hat{Q}_{XY},\hat{P}_{XY}) \in {\cal E}(\theta^\san{xy}(Q)) \times {\cal E}(\theta^\san{xy}(P))$ such that the marginals
	satisfy $\hat{Q}_X(1) = \hat{P}_X(1) = \eta^\san{x}$ and $\hat{Q}_Y(1) = \hat{P}_Y(1) = \eta^\san{y}$. 
	Then, we can numerically solve \eqref{eq:condition-lambda}, \eqref{eq:parallel-condition-1}, and \eqref{eq:parallel-condition-2} 
	with respect to $(\eta^\san{x},\eta^\san{y})$ to find $(Q_{XY}^\lambda,P_{XY}^\lambda)$.\footnote{For this binary example, we were able to find the solution
	by using builtin functions of Mathematica.}
	Fig.~\ref{Fig:comparison-1} compares the trade-off between the type I error probability and the type II error probability 
	for our proposed testing scheme with $\tau = \lambda$ and Han-Kobayashi's scheme 
	(see Theorem \ref{theorem:performance-NP-type-test} and Proposition \ref{proposition:performance-han-kobayashi}).
	The distributions are chosen to be 
	\begin{align}
	P_{XY} = \left[ 
	\begin{array}{cc}
	1/2 & 1/8 \\
	1/8 & 1/4
	\end{array}
	\right],~~~
	Q_{XY} = \left[
	\begin{array}{cc}
	1/8 & 1/4 \\
	1/2 & 1/8
	\end{array}
	\right], \label{eq:example-distribution}
	\end{align}
	and the block length is $n=100$.
	The plots in the figure indicate that,
	for a short block length such as $n=100$, our proposed testing scheme outperforms 
	the previously known scheme of Han-Kobayashi.\footnote{Numerically computing the exact trade-off for multinomial distribution
		becomes computationally intractable as the block length becomes larger.}
	
	In Fig.~\ref{Fig:trajectory}, we plotted the trajectories of $(\theta^\san{x}(P_{XY}^\lambda), \theta^\san{y}(P_{XY}^\lambda))$
	and $(\theta^\san{x}(Q_{XY}^\lambda), \theta^\san{y}(Q_{XY}^\lambda))$ by varying $\lambda$, where the distributions are the same as \eqref{eq:example-distribution}.
	For visual convenience, we inserted vectors
	$(\theta^\san{x}(P_{XY}^{\lambda}) - \theta^\san{x}(P_{XY}), \theta^\san{y}(P_{XY}^{\lambda}) - \theta^\san{y}(P_{XY}))$
	and $(\theta^\san{x}(Q_{XY}^{\lambda}) - \theta^\san{x}(Q_{XY}), \theta^\san{y}(Q_{XY}^{\lambda}) - \theta^\san{y}(Q_{XY}))$
	for two values of $\lambda$ (purple vectors and green vectors, respectively).
	As is predicted from \eqref{eq:condition-alignment} (see also Remark \ref{remark:parallel-condition}), we can verify that the vectors of the same color are parallel to each other.

%%%
\begin{figure}[!t]
\centering{
%\begin{minipage}{.5\textwidth}
\includegraphics[width=0.5\textwidth, bb=0 0 260 160]{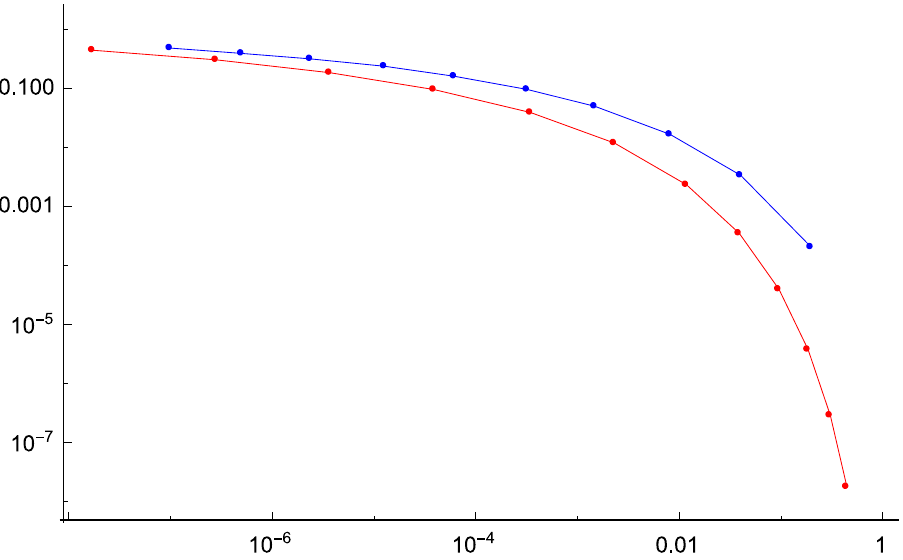}
\caption{A comparison of the trade-off between the type I error probability (horizontal axis) and the type II error probability (vertical axis)
for our proposed testing scheme (red dots)
and Han-Kobayashi's scheme (blue dots). The distributions $P_{XY}$ and 
$Q_{XY}$ are chosen to be \eqref{eq:example-distribution}, and the blocklength is $n=100$.}
\label{Fig:comparison-1}
%\end{minipage}
}
\end{figure}
%%%

%%%
\begin{figure}[!t]
\centering{
%\begin{minipage}{.4\textwidth}
\includegraphics[width=0.5\textwidth, bb=0 0 595 369]{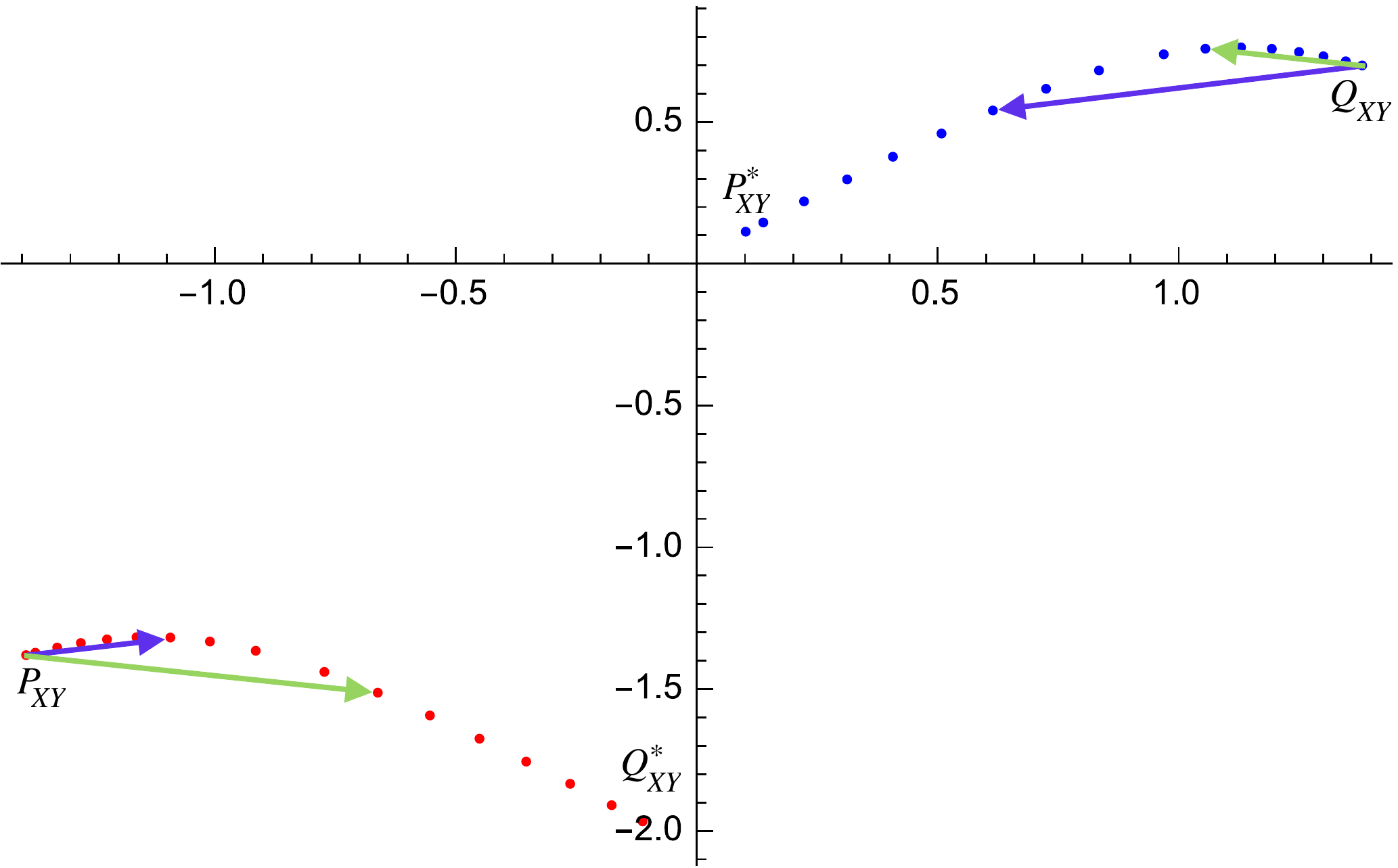}
\caption{The trajectory (red dots) of $(\theta^\san{x}(P_{XY}^\lambda), \theta^\san{y}(P_{XY}^\lambda))$ and
the trajectory (blue dots) of $(\theta^\san{x}(Q_{XY}^\lambda), \theta^\san{y}(Q_{XY}^\lambda))$ in $(\theta^\san{x},\theta^\san{y})$-plane
for varying $\lambda \in [-E(Q_{XY}\|P_{XY}), E(P_{XY}\|Q_{XY})]$.}
\label{Fig:trajectory}
%\end{minipage}
}
\end{figure}
%%%	
	%%%
	%\section{Comparison with One-Bit Scheme and Comparator?}

	%%%
	\section{Large Deviation Regime} \label{sec:LDP}
	
	In this section, we discuss the large deviation performance of our proposed testing scheme. 
	For this purpose, we need some preparations. Recall the notations introduced at the end of Section \ref{section:non-asymptotic}.
	Since the potential functions $\psi_{Q,\lambda}(s)$ and $\psi_{P,\lambda}(t)$ of the exponential families are strict convex, $\psi_{Q,\lambda}^\prime(s)$
	and $\psi_{P,\lambda}^\prime(t)$ are monotonically increasing functions.
	Thus, we can define the inverse functions $s_{Q,\lambda}(\tau)$ and $t_{P,\lambda}(\tau)$ of $\psi_{Q,\lambda}^\prime(s)$ and $\psi_{P,\lambda}^\prime(t)$ by
	\begin{align}
	\psi_{Q,\lambda}^\prime(s_{Q,\lambda}(\tau)) &= \sum_{x,y} Q^\lambda_{XY,s_{Q,\lambda}(\tau)}(x,y) \Lambda_\lambda(x,y) \\
	&= \tau, \label{eq:expectation-parameter-exponential-Q} \\
	\psi_{P,\lambda}^\prime(t_{P,\lambda}(\tau)) &= \sum_{x,y} P^\lambda_{XY,t_{P,\lambda}(\tau)}(x,y) \Lambda_\lambda(x,y) \\
	&= \tau, \label{eq:expectation-parameter-exponential-P}
	\end{align}
	respectively. Note that $s_{Q,\lambda}(\tau)$ and $t_{P,\lambda}(\tau)$ are the expectation parameters of the exponential families
	${\cal E}_\lambda(Q)$ and ${\cal E}_\lambda(P)$. 
	We also use the following expressions of the relative entropies:
	\begin{align}
	D(Q^\lambda_{XY,s_{Q,\lambda}(\tau)} \| Q_{XY}) &= s_{Q,\lambda}(\tau) \tau - \psi_{Q,\lambda}(s_{Q,\lambda}(\tau)), \label{eq:divergence-expression-1} \\
	D(P^\lambda_{XY,t_{P,\lambda}(\tau)} \| P_{XY}) &= t_{P,\lambda}(\tau) \tau - \psi_{P,\lambda}(t_{P,\lambda}(\tau)). \label{eq:divergence-expression-2}
	\end{align}

	\begin{theorem} \label{theorem:large-deviation}
		For $- E(Q_{XY} \| P_{XY}) < \lambda < E(P_{XY} \| Q_{XY})$, 
		the Neyman-Pearson-like testing scheme $T_n^\mathtt{NPl}$ of Section \ref{section:non-asymptotic} with threshold $\tau = \lambda$ satisfies 
		\begin{align}
		\alpha[T_n^\mathtt{NPl}] &\le \exp\{ -n D(P_{XY}^\lambda \| P_{XY}) \}, \label{eq:large-deviation-performance-type-I} \\
		\beta[T_n^\mathtt{NPl}] &\le \exp\{ - n D(Q_{XY}^\lambda \| Q_{XY})  \} \label{eq:large-deviation-performance-type-II} \\
		&= \exp\{ - n F(D(P_{XY}^\lambda \| P_{XY})) \} \label{eq:large-deviation-performance-type-II-2}
		\end{align}
	\end{theorem}
	\begin{proof}
		By applying the Markov inequality to Theorem \ref{theorem:performance-NP-type-test}, for any $t \le 0$ and $s \ge 0$, we have (see \cite{dembo-zeitouni-book})
		\begin{align}
		P\bigg( \frac{1}{n} \sum_{i=1}^n \Lambda_\lambda(X_i, Y_i) < \tau \bigg) \le \exp\{ - n (t \tau - \psi_{P,\lambda}(t)) \} 
		\label{eq:markov-bound-1}
		\end{align}
		and 
		\begin{align}
		Q\bigg( \frac{1}{n} \sum_{i=1}^n \Lambda_\lambda(X_i,Y_i) \ge \tau \bigg) \le \exp\{ - n (s \tau - \psi_{Q,\lambda}(s)) \}. 
		\label{eq:markov-bound-2}
		\end{align}
		
		From \eqref{eq:expectation-parameter-exponential-Q} and \eqref{eq:condition-range}, we have
		\begin{align}
		\psi_{Q,\lambda}^\prime(s_{Q,\lambda}(\lambda)) &= \sum_{x,y} Q^\lambda_{XY,s_{Q,\lambda}(\lambda)}(x,y) \Lambda_\lambda(x,y) \\
		&= \lambda \\
		&> \sum_{x,y} Q_{XY}(x,y) \Lambda_\lambda(x,y) \\
		&= \sum_{x,y} Q^\lambda_{XY,0}(x,y) \Lambda_\lambda(x,y) \\
		&= \psi_{Q,\lambda}^\prime(0),
		\end{align}
		which, together with the fact that $\psi_{Q,\lambda}^\prime(s)$ is an increasing function,
		imply $s_{Q,\lambda}(\lambda) > 0$.
		Similarly, from \eqref{eq:expectation-parameter-exponential-P} and \eqref{eq:condition-range},
		we have $t_{P,\lambda}(\lambda) < 0$. 
		
		From \eqref{eq:condition-lambda} and Corollary \ref{corollary:implication-pythagorean} combined with the same marginal conditions
		(see \eqref{eq:condition-pair-1} and \eqref{eq:condition-pair-2}), we have
		\begin{align}
		\lambda &= D(Q_{XY}^\lambda \| Q_{XY}) - D(P_{XY}^\lambda \| P_{XY}) \\
		&= \sum_{x,y} Q_{XY}^\lambda(x,y) \bigg[ \log \frac{Q_{XY}^\lambda(x,y)}{Q_{XY}(x,y)} - \log \frac{P_{XY}^\lambda(x,y)}{P_{XY}(x,y)} \bigg] \\
		&= \sum_{x,y} Q_{XY}^\lambda(x,y) \Lambda_\lambda(x,y). \label{eq:expectation-parameter-condition-Qlambda}
		\end{align}
		Furthermore, \eqref{eq:condition-alignment} together with the definition of $Q^\lambda_{XY,s}$ (see \eqref{eq:definition-exponential-family-Q-s})
		implies $Q_{XY}^\lambda \in {\cal E}_\lambda(Q)$. This together with \eqref{eq:expectation-parameter-condition-Qlambda} means that
		$Q_{XY}^\lambda = Q^\lambda_{XY,s_{Q,\lambda}(\lambda)}$. Similarly, we have $P_{XY}^\lambda = P^\lambda_{XY, t_{P,\lambda}(\lambda)}$.
		
		Consequently, by substituting $\tau = \lambda$, $t = t_{P,\lambda}(\lambda)$, and $s=s_{Q,\lambda}(\lambda)$ into \eqref{eq:markov-bound-1} and \eqref{eq:markov-bound-2}, 
		we have \eqref{eq:large-deviation-performance-type-I} and
		\eqref{eq:large-deviation-performance-type-II}  from \eqref{eq:divergence-expression-1} and  \eqref{eq:divergence-expression-2}. 
		Furthermore, by \eqref{eq:definition-lambda-r-conversion} and \eqref{eq:condition-lambda} for $r = D(P_{XY}^\lambda \| P_{XY})$, we have \eqref{eq:large-deviation-performance-type-II-2}.
	\end{proof}
	
	%By noting \eqref{eq:condition-pair-1} and \eqref{eq:condition-pair-2}, we find that 
	%\begin{align}
	%Q_{XY}^\lambda \in \{ \tilde{Q}_{XY} : E(\tilde{Q}_{XY} \| P_{XY}) \le D(P_{XY}^\lambda \| P_{XY}) \},
	%\end{align}
	%which implies $D(Q^\lambda_{XY} \| Q_{XY}) \ge F(D(P_{XY}^\lambda \| P_{XY}))$.
	Note that Theorem \ref{theorem:large-deviation} means that 
	our Neyman-Pearson-like testing scheme with threshold $\tau = \lambda$ is optimal in the large deviation regime.
	Compared to the derivation of the same exponents based on the method of type, 
	Theorem \ref{theorem:large-deviation} has the advantage in that there is no polynomial factor of $n$ 
	that stems from the number of types.
	
	In Fig.~\ref{Fig:comparison-exponent}, we plotted the trade-off of the two exponents
	\begin{align}
	\big(D(P_{XY}^\lambda \| P_{XY}), F(D(P_{XY}^\lambda \| P_{XY})) \big) = \big( D(P_{XY}^\lambda \| P_{XY}), D(Q_{XY}^\lambda \| Q_{XY}) \big)
	\label{eq:optimal-exponent-trade-off}
	\end{align}
	by varying $\lambda \in [-E(Q_{XY}\|P_{XY}), E(P_{XY}\|Q_{XY})]$.
	For fixed values of $\lambda$, say $\lambda = \overline{\lambda} = E(P_{XY}\|Q_{XY})$
	or $\lambda = \underline{\lambda} = - E(Q_{XY}\|P_{XY})$, we can also achieve the following trade-offs 
	by using our testing scheme:\footnote{To prove the achievability of \eqref{eq:exponent-trade-off-upper-lambda}
		(or \eqref{eq:exponent-trade-off-lower-lambda}), we take $\lambda = \overline{\lambda}$ (or $\lambda = \underline{\lambda}$),
		$t = t_{P,\lambda}(\tau)$, and $s = s_{Q,\lambda}(\tau)$ in \eqref{eq:markov-bound-1} and \eqref{eq:markov-bound-2}.}
	\begin{align}
	\big( D(P_{XY,t_{P,\overline{\lambda}}(\tau)}^{\overline{\lambda}} \| P_{XY}), D(Q_{XY, s_{Q,\overline{\lambda}}(\tau)}^{\overline{\lambda}} \| Q_{XY}) \big)
	\label{eq:exponent-trade-off-upper-lambda}
	\end{align}
	for $\sum_{x,y} Q_{XY}(x,y) \Lambda_{\overline{\lambda}}(x,y) \le \tau \le E(P_{XY}\|Q_{XY})$ and
	\begin{align}
	\big( D(P_{XY,t_{P,\underline{\lambda}}(\tau)}^{\underline{\lambda}} \| P_{XY}), D(Q_{XY, s_{Q,\underline{\lambda}}(\tau)}^{\underline{\lambda}} \| Q_{XY}) \big)
	\label{eq:exponent-trade-off-lower-lambda}
	\end{align}
	for $- E(Q_{XY} \| P_{XY}) \le \tau \le \sum_{x,y} P_{XY}(x,y) \Lambda_{\underline{\lambda}}(x,y)$.
	For comparison, we also plotted these trade-offs in Fig.~\ref{Fig:comparison-exponent}.
	The adjustment of $\lambda$ is crucial to achieve the optimal trade-off; 
	when either the type I exponent or the type II exponent is very small, then $\lambda=\overline{\lambda}$ or
	$\lambda=\underline{\lambda}$ is quite effective.
	
%%%
\begin{figure}[!t]
\centering{
%\begin{minipage}{.5\textwidth}
\includegraphics[width=0.5\textwidth, bb=0 0 260 160]{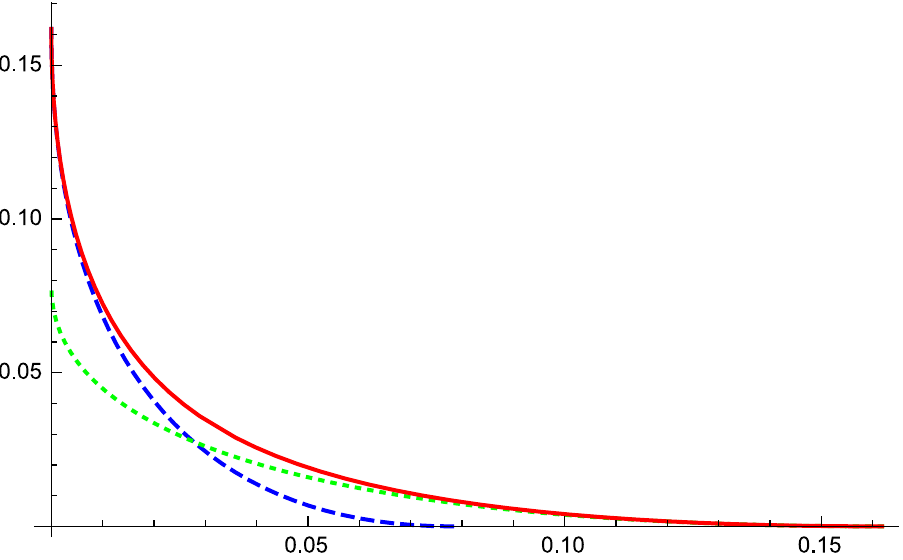}
\caption{A comparison of the trade-offs between the type I and type II exponent, where the horizontal axis is
type I exponent and the vertical axis is type II exponent.
The red solid curve is the optimal trade-off between the type I exponent and the type II exponent, i.e., \eqref{eq:optimal-exponent-trade-off}; the blue dashed curve is 
the trade-off between the type I exponent and type II exponent
for $\lambda=\overline{\lambda}$, i.e., \eqref{eq:exponent-trade-off-upper-lambda}; 
the green dotted curve is the trade-off 
between the type I exponent and the type II exponent
for $\lambda = \underline{\lambda}$, i.e., \eqref{eq:exponent-trade-off-lower-lambda}.
The distributions $P_{XY}$ and 
$Q_{XY}$ are chosen to be \eqref{eq:example-distribution}.}
\label{Fig:comparison-exponent}
%\end{minipage}
}
\end{figure}
%%%

	%%%%%%%%%%%%%%%%%%%%%%%%%%%%%%%%%%%%%%%%%%%
	\section{Second-Order Analysis} \label{sec:second-order}
	
	Let us start by introducing the projected relative entropy density:
	\begin{align} \label{eq:definition-density}
	\jmath_{P\|Q}(x,y) := \log \frac{P_{XY}^*(x,y)}{Q_{XY}(x,y)},
	\end{align}
	where $P_{XY}^*$ is the optimizer of $E(P_{XY}\|Q_{XY})$. 
	From Corollary \ref{corollary:implication-pythagorean},
	we have the following proposition, which justifies  
	referring to $\jmath_{P\|Q}$ as the projected relative entropy density. 

	\begin{proposition} \label{proposition:validity-of-density}
		It holds that 
		\begin{align}
		\san{E}\big[ \jmath_{P\|Q}(X,Y) \big] = E(P_{XY} \| Q_{XY}), 
		\end{align}
		where the expectation is taken over $(X,Y) \sim P_{XY}$. 
	\end{proposition}

	Proposition \ref{proposition:first-order} shows that the optimal exponent of the type II error probability is given by
	$E(P_{XY}\|Q_{XY})$, and it can be achieved by a symmetric scheme.
	In this section, for a given constraint $0 < \varepsilon < 1$ on the type I error probability, 
	we consider the optimal second-order exponent among the class of symmetric schemes:
	%\begin{align}
	%F_0(\varepsilon) := \sup\left\{ \liminf_{n\to\infty} \frac{-\log\beta[T_n] - n E^*(P_{XY})}{\sqrt{n}} : \{T_n\}_{n=1}^\infty \mbox{ is zero-rate},\limsup_{n\to\infty} \alpha[T_n] \le \varepsilon \right\}
	%\end{align}
	%and
	\begin{align}
	G_\san{s}(\varepsilon) := \sup\left\{ \liminf_{n\to\infty} \frac{-\log\beta[T_n] - n E(P_{XY}\|Q_{XY})}{\sqrt{n}} : \{T_n\}_{n=1}^\infty \mbox{ is symmetric},\limsup_{n\to\infty} \alpha[T_n] \le \varepsilon \right\}.
	\label{eq:definition-second-order-exponent}
	\end{align}
	%Again by definition, $F_0(\varepsilon) \ge F_\san{s}(\varepsilon)$.
	In approximate terms, this means that the type I error probability and type II error probability behave as follows:
	\begin{align}
	\alpha[T_n] \simeq \varepsilon,~~~\beta[T_n] \simeq \exp\{- n E(P_{XY}\|Q_{XY}) - \sqrt{n} G_{\san{s}}(\varepsilon) \}. 
	\end{align}
	The following theorem characterizes the second-order coefficient $G_{\san{s}}(\varepsilon)$.

	\begin{theorem} \label{theorem:second-order}
		For a given $0 < \varepsilon < 1$, it holds that
		\begin{align}
		G_{\san{s}}(\varepsilon) = \sqrt{ V(P\|Q) } \Phi^{-1}(\varepsilon),
		\end{align}
		where 
		\begin{align}
		V(P\|Q) := \san{Var}\big[ \jmath_{P\|Q}(X,Y) \big]
		\end{align}
		for $(X,Y) \sim P_{XY}$.
	\end{theorem}
	
	\begin{remark} \label{remark:third-order-bound}
		In fact, we can present a slightly stronger statement than Theorem \ref{theorem:second-order} for the achievability, i.e., 
		the Neyman-Pearson-like testing scheme $T_n^\mathtt{NPl}$ of Section \ref{section:non-asymptotic} performs as follows for sufficiently large $n$:
		\begin{align}
		\alpha[T_n^\mathtt{NPl}] \le \varepsilon
		\end{align}
		and
		\begin{align}
		- \log \beta[T_n^\mathtt{NPl}] \ge n E(P_{XY} \| Q_{XY}) + \sqrt{n V(P\|Q)} \Phi^{-1}(\varepsilon) + \frac{1}{2} \log n + O(1). 
		\end{align}
	\end{remark}
	
	%%%%%%%%%%%
	\subsection{Proof of Achievability of Theorem \ref{theorem:second-order}}
	
	We prove the stronger statement, i.e., Remark \ref{remark:third-order-bound}.
	We use the following technical lemma shown in \cite[Lemma 47]{polyanskiy:10}.

	\begin{lemma} \label{lemma:PPV}
		Let $Z_1,\ldots,Z_n$ be i.i.d. random variables with $\sigma^2 = \san{Var}[Z_i]$,
		$T= \san{E}[|Z_i - \san{E}[Z_i]|^3]$. Then, for any $\gamma$, 
		\begin{align}
		\san{E}\left[ \exp \bigg( - \sum_{i=1}^n Z_i \bigg) \mathbf{1}\bigg\{ \sum_{i=1}^n Z_i > \gamma \bigg\} \right]
		\le 2 \bigg( \frac{\log 2}{\sqrt{2\pi}} + \frac{12 T}{\sigma^2} \bigg) \frac{\exp(-\gamma)}{\sigma \sqrt{n}},
		\end{align}
		where $\mathbf{1}\{ \cdot \}$ is the indicator function.
	\end{lemma}
	
	From Theorem \ref{theorem:performance-NP-type-test} with $\lambda = E(P_{XY} \| Q_{XY})$,\footnote{Note that $\Lambda_\lambda(x,y) = \jmath_{P\|Q}(x,y)$
		for $\lambda = E(P_{XY}\|Q_{XY})$.} 
	the testing scheme of Section \ref{section:non-asymptotic} satisfies 
	\begin{align}
	\beta[T_n^\mathtt{NPl}] &= Q\bigg( \frac{1}{n} \sum_{i=1}^n \jmath_{P\|Q}(X_i, Y_i) > \tau \bigg) \\
	&= \sum_{\bm{x},\bm{y}} Q_{XY}^n(\bm{x},\bm{y}) \mathbf{1}\bigg\{ \log \frac{P_{XY}^{* n}(\bm{x},\bm{y})}{Q_{XY}^n(\bm{x},\bm{y})} > \tau n \bigg\} \\
	&= \sum_{\bm{x},\bm{y}} P_{XY}^{* n}(\bm{x},\bm{y}) \exp\bigg( - \log \frac{P_{XY}^{* n}(\bm{x},\bm{y})}{Q_{XY}^n(\bm{x},\bm{y})} \bigg)
	\mathbf{1}\bigg\{ \log \frac{P_{XY}^{* n}(\bm{x},\bm{y})}{Q_{XY}^n(\bm{x},\bm{y})} > \tau n \bigg\} \\
	&\le 2 \bigg( \frac{\log 2}{\sqrt{2\pi}} + \frac{12 T}{\sigma^2} \bigg) \frac{\exp( - \tau n)}{\sigma \sqrt{n}},  
	\label{eq:bound-on-type-2-error}
	\end{align}
	where we used Lemma \ref{lemma:PPV} by setting $Z_i = \jmath_{P\|Q}(X_i,Y_i)$ for $(X_i,Y_i) \sim P_{XY}^*$
	in the last inequality. Now, we set
	\begin{align} 
	\tau = E(P_{XY} \| Q_{XY}) + \sqrt{\frac{V(P\|Q)}{n}} \Phi^{-1}\bigg( \varepsilon - \frac{6 T(P\|Q)}{\sqrt{n} V(P\|Q)^{3/2}} \bigg),
	\label{eq:setting-lambda}
	\end{align}
	where $T(P\|Q)$ is the absolute third moment of $\jmath_{P\|Q}(X,Y)$ for $(X,Y) \sim P_{XY}$. Then, application
	of the Berry-Ess\'een theorem (see \cite{feller:book}) yields 
	\begin{align}
	\alpha[T_n^\mathtt{NPl}] &= P\bigg( \frac{1}{n} \sum_{i=1}^n \jmath_{P\|Q}(X_i, Y_i) \le \tau \bigg) \\
	&\le \varepsilon.
	\end{align}
	On the other hand, \eqref{eq:bound-on-type-2-error}, \eqref{eq:setting-lambda}, and the Taylor approximation of
	$\Phi^{-1}(\cdot)$ around $\varepsilon$ yields 
	\begin{align}
	- \log \beta[T_n^\mathtt{NPl}] &\ge \tau n + \log (\sigma \sqrt{n}) + O(1) \\
	&= n E(P_{XY} \| Q_{XY}) + \sqrt{ n V(P\|Q)} \Phi^{-1}(\varepsilon) + \frac{1}{2} \log n  + O(1),
	\end{align}
	which completes the achievability proof.

	%%%%%%%%%%%
	\subsection{Proof of Converse of Theorem \ref{theorem:second-order}}
	
	First, we provide a simple converse bound for the class of symmetric schemes.  
	For a given joint type $P_{\bar{X}\bar{Y}} \in {\cal P}_n({\cal X}\times {\cal Y})$, let
	\begin{align}
	E_n(P_{\bar{X}\bar{Y}} \| Q_{XY}) := \min\left\{ D(\tilde{P}_{XY} \| Q_{XY}) : \tilde{P}_{XY} \in {\cal P}_n({\cal X}\times{\cal Y}), \tilde{P}_X=P_{\bar{X}}, \tilde{P}_Y=P_{\bar{Y}} \right\}.
	\end{align}
	By definition, $E_n(P_{\bar{X}\bar{Y}} \| Q_{XY}) \ge E(P_{\bar{X}\bar{Y}} \| Q_{XY})$ for any $P_{\bar{X}\bar{Y}} \in {\cal P}_n({\cal X}\times {\cal Y})$.
	
	\begin{proposition} \label{proposition:converse}
		For any $r > 0$ and for any symmetric scheme $T_n$ such that 
		\begin{align} \label{eq:assumption-finite-converse}
		\beta[T_n] \le \exp\{- r n\},
		\end{align}
		it holds that 
		\begin{align} \label{eq:consequence-finite-converse}
		\alpha[T_n] \ge P\bigg( E_n(\san{t}_{X^n Y^n} \| Q_{XY}) < r - (|{\cal X}||{\cal Y}|-1)\frac{ \log (n+1)}{n} \bigg).
		\end{align}
	\end{proposition}
	\begin{proof}
		Since $T_n$ is a symmetric scheme, without loss of generality, we can assume that 
		the decoder $g$ only depends on marginal types $(P_{\bar{X}}, P_{\bar{Y}}) \in {\cal P}_n({\cal X}) \times {\cal P}_n({\cal Y})$.
		We claim that, if 
		\begin{align} 
		E_n( P_{\bar{X}} \times P_{\bar{Y}} \| Q_{XY}) < r - (|{\cal X}||{\cal Y}|-1) \frac{\log (n+1)}{n}, 
		\end{align}
		then $g(P_{\bar{X}}, P_{\bar{Y}}) = \san{H}_1$. Otherwise, by letting $P_{\tilde{X}\tilde{Y}} \in {\cal P}_n({\cal X}\times {\cal Y})$ be
		such that $P_{\tilde{X}} = P_{\bar{X}}$, $P_{\tilde{Y}} = P_{\bar{Y}}$, and $E_n(P_{\bar{X}} \times P_{\bar{Y}} \|Q_{XY}) = D(P_{\tilde{X}\tilde{Y}} \| Q_{XY})$, we have
		\begin{align}
		\beta[T_n]
		&\ge Q_{XY}^n({\cal T}_{\tilde{X}\tilde{Y}}^n) \\
		&\ge \frac{1}{(n+1)^{(|{\cal X}||{\cal Y}|-1)}} \exp\{- n D(P_{\tilde{X}\tilde{Y}} \| Q_{XY}) \} \\
		&= \frac{1}{(n+1)^{(|{\cal X}||{\cal Y}|-1)}} \exp\{- n E_n(P_{\bar{X}} \times P_{\bar{Y}} \| Q_{XY}) \} \\
		&> \exp\{- r n \},
		\end{align}
		which contradict \eqref{eq:assumption-finite-converse}, where ${\cal T}_{\tilde{X}\tilde{Y}}^n$ is the set of all pairs $(\bm{x},\bm{y})$
		such that $\san{t}_{\bm{x}\bm{y}} = P_{\tilde{X}\tilde{Y}}$. 
		Thus, we have \eqref{eq:consequence-finite-converse}.
	\end{proof}

	Second, we approximate $E_n(P_{\bar{X}\bar{Y}} \|Q_{XY})$ by $E(P_{\bar{X}\bar{Y}} \| Q_{XY})$. 
	
	\begin{lemma} \label{lemma:approximation-by-type}
		For any $P_{\bar{X}\bar{Y}} \in {\cal P}_n({\cal X}\times {\cal Y})$, it holds that 
		\begin{align}
		E_n(P_{\bar{X}\bar{Y}} \|Q_{XY})  \le  E(P_{\bar{X}\bar{Y}} \| Q_{XY}) +  \Delta_n,
		\end{align}
		where 
		\begin{align}
		\Delta_n := \nu_n  \log \frac{|{\cal X}||{\cal Y}|}{\nu_n}
		+ \nu_n \max_{x,y} \log \frac{1}{Q_{XY}(x,y)}.
		\end{align}
		and
		\begin{align}
		\nu_n := \frac{4(|{\cal X}|-1)(|{\cal Y}|-1)}{n}.
		\end{align}
	\end{lemma}
	%%%
	\begin{proof}
		Let $P_{\bar{X}\bar{Y}}^*$ be such that $P^*_{\bar{X}} = P_{\bar{X}}$, $P^*_{\bar{Y}} = P_{\bar{Y}}$, and 
		\begin{align}
		D(P^*_{\bar{X}\bar{Y}} \| Q_{XY}) = E(P_{\bar{X}\bar{Y}} \|Q_{XY}).
		\end{align}
		We set $\tilde{P}_{XY} \in {\cal P}_n({\cal X}\times {\cal Y})$ as follows:
		$\tilde{P}_X = P^*_{\bar{X}}$, $\tilde{P}_Y = P^*_{\bar{Y}}$, and 
		\begin{align}
		\tilde{P}_{XY}(x,y) = \big\lfloor P^*_{\bar{X}\bar{Y}}(x,y) \big\rfloor_{1/n}
		\end{align}
		for $x=1,\ldots,|{\cal X}|-1$ and $y=1,\ldots,|{\cal Y}|-1$,
		where
		\begin{align}
		\lfloor t \rfloor_{1/n} := \max\left\{ \frac{k}{n} : k \in \mathbb{N}, \frac{k}{n} \le t \right\}.
		\end{align}
		Then, we have
		\begin{align}
		&\| \tilde{P}_{XY} - P^*_{\bar{X}\bar{Y}} \| \\
		&= \bigg| \bigg[ \tilde{P}_X(0) + \tilde{P}_Y(0) -1 + \sum_{x=1}^{|{\cal X}|-1} \sum_{y=1}^{|{\cal Y}|-1} \tilde{P}_{XY}(x,y) \bigg] 
		- \bigg[P^*_{\bar{X}}(0) + P^*_{\bar{Y}}(0) -1 + \sum_{x=1}^{|{\cal X}|-1} \sum_{y=1}^{|{\cal Y}|-1} P^*_{\bar{X}\bar{Y}}(x,y) \bigg]  \bigg| \\
		&~~~+ \sum_{x=1}^{|{\cal X}|-1} \bigg| \bigg[ \tilde{P}_X(x) - \sum_{y=1}^{|{\cal Y}|-1} \tilde{P}_{XY}(x,y) \bigg] 
		- \bigg[ P^*_{\bar{X}}(x) - \sum_{y=1}^{|{\cal Y}|-1} P^*_{\bar{X}\bar{Y}}(x,y) \bigg] \bigg| \\
		&~~~+ \sum_{y=1}^{|{\cal Y}|-1} \bigg| \bigg[ \tilde{P}_Y(y) - \sum_{x=1}^{|{\cal X}|-1} \tilde{P}_{XY}(x,y) \bigg] 
		- \bigg[ P^*_{\bar{Y}}(y) - \sum_{x=1}^{|{\cal X}|-1} P^*_{\bar{X}\bar{Y}}(x,y) \bigg] \bigg| \\
		&~~~+ \sum_{x=1}^{|{\cal X}|-1} \sum_{y=1}^{|{\cal Y}|-1} \bigg| \tilde{P}_{XY}(x,y) - P^*_{\bar{X}\bar{Y}}(x,y) \bigg| \\
		&\le \nu_n.
		\end{align}
		Thus, by the continuity of the entropy \cite[Lemma 2.7]{csiszar-korner:11}, we have
		\begin{align}
		E_n(P_{\bar{X}\bar{Y}} \| Q_{XY}) 
		&\le D(\tilde{P}_{XY} \| Q_{XY}) \\
		&= - H(\tilde{P}_{XY}) + \sum_{x,y} \tilde{P}_{XY}(x,y) \log \frac{1}{Q_{XY}(x,y)} \\
		&\le - H(P^*_{\bar{X}\bar{Y}}) + \nu_n \log \frac{|{\cal X}||{\cal Y}|}{\nu_n }
		+ \sum_{x,y} P^*_{\bar{X}\bar{Y}}(x,y) \log \frac{1}{Q_{XY}(x,y)} 
		+\nu_n \max_{x,y} \log \frac{1}{Q_{XY}(x,y)} \\
		&= D(P^*_{\bar{X}\bar{Y}} \| Q_{XY}) + \Delta_n.
		\end{align}
	\end{proof}
	
	Next, we show a useful relationship between the projected relative entropy $E(P_{XY}\|Q_{XY})$ 
	and its density $\jmath_{P\|Q}$.
	\begin{lemma} \label{lemma:derivative-of-projected-relative-entropy}
		Let $\tilde{\eta}$ be the expectation parameter of $P_{XY}$ (see Section \ref{sec:information-geometry}). Then, it holds that 
		\begin{align}
		\frac{\partial E(P_{XY,\eta} \| Q_{XY})}{\partial \eta^{\san{x}}_i} \bigg|_{\eta = \tilde{\eta}} &= \jmath_{P\|Q}(i,0) - \jmath_{P\|Q}(0,0), \label{eq:derivative-1} \\
		\frac{\partial E(P_{XY,\eta} \| Q_{XY})}{\partial \eta^{\san{y}}_j} \bigg|_{\eta = \tilde{\eta}} &= \jmath_{P\|Q}(0,j) - \jmath_{P\|Q}(0,0), \label{eq:derivative-2} \\
		\frac{\partial E(P_{XY,\eta} \| Q_{XY})}{\partial \eta^{\san{xy}}_{ij}} \bigg|_{\eta = \tilde{\eta}} &= 0 \label{eq:derivative-3}
		\end{align}
		for $1 \le i \le d_{\san{x}}$ and $1 \le j \le d_{\san{y}}$, where $\tilde{\eta} = \eta(P)$ is the expectation parameter of $P_{XY}$.
	\end{lemma}
	%%%
	\begin{proof}
		We first show \eqref{eq:derivative-3}. In fact, this follows from the fact that 
		\begin{align}
		E(P_{XY,\tilde{\eta} + \Delta} \| Q_{XY}) =  E(P_{XY,\tilde{\eta}} \| Q_{XY})
		\end{align} 
		for every vector $\Delta$ such that the values are $0$ except for the $\eta^{\san{xy}}_{ij}$-coordinate.
		
		We prove \eqref{eq:derivative-1} by first showing that 
		\begin{align}
		\frac{\partial E(P_{XY,\eta} \| Q_{XY})}{\partial \eta^{\san{x}}_i} \bigg|_{\eta=\tilde{\eta}} = \frac{\partial E(P_{XY,\eta} \| Q_{XY})}{\partial \eta^{\san{x}}_i} \bigg|_{\eta=\hat{\eta}}, 
		\end{align}
		where $\hat{\eta}$ is the expectation parameter of $P_{XY}^*$.
		Since the $(\eta^\san{x},\eta^\san{y})$-coordinate of $\tilde{\eta}$ and $\hat{\eta}$ are the same, note that 
		\begin{align}
		E(P_{XY,\tilde{\eta}+\Delta} \|Q_{XY}) = E(P_{XY,\hat{\eta}+\Delta} \|Q_{XY})
		\end{align}
		holds for every vector $\Delta$ such that the values are $0$ with the exception of the $\eta^{\san{x}}_i$-coordinate.
		Let $\Delta \eta^{\san{x}}_i$ be the $\eta^{\san{x}}_i$-coordinate of $\Delta$. Then, we have
		\begin{align}
		\frac{\partial E(P_{XY,\eta} \| Q_{XY})}{\partial \eta^{\san{x}}_i} \bigg|_{\eta=\tilde{\eta}}
		&= \lim_{\Delta \eta^{\san{x}}_i \to 0} \frac{E(P_{XY,\tilde{\eta}+\Delta} \|Q_{XY}) - E(P_{XY,\tilde{\eta}} \|Q_{XY})}{\Delta \eta^{\san{x}}_i} \\
		&= \lim_{\Delta \eta^{\san{x}}_i \to 0} \frac{E(P_{XY,\hat{\eta}+\Delta} \|Q_{XY}) - E(P_{XY,\hat{\eta}} \|Q_{XY})}{\Delta \eta^{\san{x}}_i} \\
		&= \frac{\partial E(P_{XY,\eta} \| Q_{XY})}{\partial \eta^{\san{x}}_i} \bigg|_{\eta=\hat{\eta}}.
		\end{align}
		
		Now, note that (see Proposition \ref{proposition:validity-of-density})
		\begin{align}
		E(P_{XY,\eta} \| Q_{XY}) = \sum_{x,y} P_{XY,\eta}(x,y) \log \frac{P_{XY,\eta}^*(x,y)}{Q_{XY}(x,y)},
		\end{align}
		where $P_{XY,\eta}^*$ is the optimizer of $E(P_{XY,\eta} \| Q_{XY})$. Furthermore, from \eqref{eq:expansion-of-expectation-parameter}, we have
		\begin{align} \label{eq:derivative-of-expectation-parameter-expansion}
		\frac{\partial P_{XY,\eta}(x,y)}{\partial \eta^\san{x}_i} = \bigg( \delta_i(x) - \delta_0(x) \bigg) \delta_0(y).
		\end{align} 
		Thus, we have
		\begin{align}
		\frac{\partial E(P_{XY,\eta} \| Q_{XY})}{\partial \eta^{\san{x}}_i} \bigg|_{\eta=\hat{\eta}} 
		&= \sum_{x,y} \frac{\partial P_{XY,\eta}(x,y)}{\partial \eta^\san{x}_i} \bigg|_{\eta = \hat{\eta}} \log \frac{P_{XY,\hat{\eta}}^*(x,y)}{Q_{XY}(x,y)} \\
		&~~~+ \sum_{x,y} P_{XY,\hat{\eta}}(x,y) \frac{\partial}{\partial \eta^\san{x}_i} \bigg[ \log \frac{P_{XY,\eta}^*(x,y)}{Q_{XY}(x,y)} \bigg] \bigg|_{\eta=\hat{\eta}} \\
		&= \log \frac{P_{XY,\hat{\eta}}^*(i,0)}{Q_{XY}(i,0)} - \log \frac{P_{XY,\hat{\eta}}^*(0,0)}{Q_{XY}(0,0)} \\
		&~~~+\sum_{x,y} \frac{\partial}{\partial \eta^\san{x}_i} P_{XY,\hat{\eta}}^*(x,y) \bigg|_{\eta=\hat{\eta}} \\
		&= \log \frac{P_{XY,\hat{\eta}}^*(i,0)}{Q_{XY}(i,0)} - \log \frac{P_{XY,\hat{\eta}}^*(0,0)}{Q_{XY}(0,0)} \\ 
		&= \jmath_{P\|Q}(i,0) - \jmath_{P\|Q}(0,0),
		\end{align}
		where the second equality follows from \eqref{eq:derivative-of-expectation-parameter-expansion}
		and $P_{XY,\hat{\eta}}^* = P_{XY,\hat{\eta}}$,
		the third equality follows from $\frac{\partial}{\partial \eta^\san{x}_i} 1 = 0$, 
		and the last equality follows from the fact that $P_{XY,\hat{\eta}}^* = P_{XY}^*$.
		Finally, \eqref{eq:derivative-2} is proved in a similar manner.
	\end{proof}
	
	From Lemma \ref{lemma:derivative-of-projected-relative-entropy}, we have the following approximation of $E(P_{\bar{X}\bar{Y}} \| Q_{XY})$
	around $P_{XY}$.
	\begin{theorem} \label{theorem:Taylor-approximation}
		It holds that
		\begin{align}
		E(P_{\bar{X}\bar{Y}} \| Q_{XY}) = \sum_{x,y} P_{\bar{X}\bar{Y}}(x,y) \jmath_{P\|Q}(x,y) + o(\|P_{\bar{X}\bar{Y}} - P_{XY} \|).
		\end{align}
	\end{theorem}
	\begin{proof}
		From Lemma \ref{lemma:derivative-of-projected-relative-entropy}, by noting that $(\eta^\san{x},\eta^\san{y})$-coordinates correspond to the marginal distributions, 
		the first-order Taylor approximation of $E(P_{\bar{X}\bar{Y}} \| Q_{XY})$
		around $P_{XY}$ is given by
		\begin{align}
		E(P_{\bar{X}\bar{Y}} \| Q_{XY}) 
		&= E(P_{XY} \| Q_{XY}) + \sum_{i=1}^{d_{\san{x}}} \big( P_{\bar{X}}(i) - P_X(i) \big) \big( \jmath_{P\|Q}(i,0) - \jmath_{P\|Q}(0,0) \big) \\
		&~~~ + \sum_{j=1}^{d_{\san{y}}} \big( P_{\bar{Y}}(j) - P_Y(j) \big) \big( \jmath_{P\|Q}(0,j) - \jmath_{P\|Q}(0,0) \big) + o(\|P_{\bar{X}\bar{Y}} - P_{XY} \|) \\
		&= E(P_{XY} \| Q_{XY}) + \sum_x \big( P_{\bar{X}}(x) - P_X(x) \big) \jmath_{P\|Q}(x,0) \\
		&~~~ + \sum_y \big( P_{\bar{Y}}(y) - P_Y(y) \big) \jmath_{P\|Q}(0,y) + o(\|P_{\bar{X}\bar{Y}} - P_{XY} \|) \\
		&= \sum_{x,y} P_{\bar{X}\bar{Y}}(x,y) \jmath_{P\|Q}(x,y) + o(\|P_{\bar{X}\bar{Y}} - P_{XY} \|),
		\end{align}
		where the last equality follows from Remark \ref{remark:crucial-identity}.
	\end{proof}
	
	Now, we are ready to prove the converse part of Theorem \ref{theorem:second-order}. Let $G$ be such that 
	\begin{align} \label{eq:assumtption-contradiction}
	G > \sqrt{V(P\|Q)} \Phi^{-1}(\varepsilon),
	\end{align}
	and let 
	\begin{align}
	\tau = E(P_{XY} \| Q_{XY}) + \frac{F}{\sqrt{n}} + (|{\cal X}||{\cal Y}|-1) \frac{\log (n+1)}{n} + \Delta_n,
	\end{align}
	where $\Delta_n$ is the residual specified in Lemma \ref{lemma:approximation-by-type}. 
	Then, Proposition \ref{proposition:converse} and Lemma \ref{lemma:approximation-by-type} imply that,
	for any symmetric scheme satisfying $\beta[T_n] \le \exp(-\tau n)$, 
	\begin{align}
	\alpha[ T_n ] \ge P\bigg( E(\san{t}_{X^n Y^n} \| Q_{XY}) < E(P_{XY} \| Q_{XY}) + \frac{G}{\sqrt{n}} \bigg).
	\end{align}
	Let 
	\begin{align}
	{\cal K}_n := \left\{ P_{\bar{X}\bar{Y}} \in {\cal P}_n({\cal X}\times {\cal Y}) : |P_{\bar{X}\bar{Y}}(x,y) - P_{XY}(x,y) | \le \sqrt{\frac{\log n}{n}}~~~\forall (x,y) \right\}.
	\end{align}
	Then, by the Hoeffding inequality, we have
	\begin{align}
	P\bigg( \san{t}_{X^n Y^n} \notin {\cal K}_n \bigg) \le \frac{2 |{\cal X}| |{\cal Y}|}{n^2}.
	\end{align}
	Thus, by using Theorem \ref{theorem:Taylor-approximation} for $P_{\bar{X}\bar{Y}} \in {\cal K}_n$, we have
	\begin{align}
	\alpha[T_n] &\ge P\bigg( \san{t}_{X^n Y^n} \in {\cal K}_n,~ E(\san{t}_{X^n Y^n} \| Q_{XY}) < E(P_{XY} \| Q_{XY}) + \frac{G}{\sqrt{n}} \bigg) \\
	&\ge P\bigg( \san{t}_{X^n Y^n} \in {\cal K}_n,~ \sum_{x,y} \san{t}_{X^n Y^n}(x,y) \jmath_{P\|Q}(x,y) < E(P_{XY} \| Q_{XY}) + \frac{G}{\sqrt{n}} - c \frac{\log n}{n} \bigg) \\
	&\ge P\bigg( \frac{1}{n} \sum_{i=1}^n \jmath_{P\|Q}(X_i,Y_i) < E(P_{XY} \| Q_{XY}) + \frac{G}{\sqrt{n}} - c \frac{\log n}{n} \bigg) - P\bigg( \san{t}_{X^n Y^n} \notin {\cal K}_n \bigg) \\
	&\ge P\bigg( \frac{1}{n} \sum_{i=1}^n \jmath_{P\|Q}(X_i,Y_i) < E(P_{XY} \| Q_{XY}) + \frac{G}{\sqrt{n}} - c \frac{\log n}{n} \bigg) - \frac{2 |{\cal X}||{\cal Y}|}{n^2}
	\end{align}
	for some constant $c >0$. Thus, by the central limit theorem, we have
	\begin{align}
	\liminf_{n\to \infty} \alpha[T_n] > \varepsilon,
	\end{align}
	which together with \eqref{eq:assumtption-contradiction} implies 
	\begin{align}
	G_{\san{s}}(\varepsilon) \le \sqrt{V(P\|Q)} \Phi^{-1}(\varepsilon). 
	\end{align}

	%%%%%%%%%%%%%%%%%%%%%%%%%%%%%%%%%%%%%%%%%%%
	\section{Conclusion and Discussion} \label{sec:conclusion}
	
	In this paper, we proposed a novel testing scheme for the zero-rate multiterminal hypothesis testing problem.
	The previously known scheme of Han-Kobayashi is based on a cylinder with respect to the relative entropy, 
	and thus their scheme can be regarded as a multiterminal analogue of the Hoeffding test.
	In contrast, our proposed scheme bisects the joint probability simplex by an appropriate mixture family, an approach reminiscent of the Neyman-Pearson test. 
	For a short block length, we numerically determined that the performance of our proposed scheme is superior to that of the previously reported scheme. 
	We also showed that, in a large deviation regime, our proposed scheme optimizes the trade-off of exponents 
	that was shown by Han-Kobayashi.
	Furthermore, we derived the optimal second-order exponent
	among the class of symmetric schemes, which can be achieved by our proposed scheme. 
	
	More ambitious goals would be to identify the optimal second-order exponent or to derive non-asymptotic 
	bounds for the class of general zero-rate testing schemes. However, such analyses are not only
	technically difficult, but they also have subtlety in the problem formulations.  
	For instance, a straightforward definition of the optimal second-order exponent by replacing
	``symmetric" with ``zero-rate" in \eqref{eq:definition-second-order-exponent} would make no sense.
	In fact, suppose that the marginals of $P_{XY}$ and $Q_{XY}$ are the same as an extreme case. 
	In that case, it is known that the optimal first-order exponent achievable by zero-rate schemes 
	is $0$. However, a trivial zero-rate scheme allows the encoders to send the first $\lceil n^\gamma \rceil$ symbols of their observations for
	a given $\frac{1}{2} <\gamma <1$; then the type II error probability behaves as $\exp\{- c n^\gamma \}$ for some constant $c>0$.
	In other words, the order of the second-order rate may depend on the growth rate of the message sizes even if they are zero-rate. 
	Identifying appropriate formulations and studying these problems would be important in future.

	%%%%%%%%%%%%%%%%%%%%%%%%%%%
	\appendix
	
	\subsection{Proof of Remark \ref{remark:crucial-identity}} \label{appendix:proof-of-remark:crucial-identity}
	
	We only prove the statement for $- E(Q_{XY}\|P_{XY}) < \lambda < E(P_{XY} \| Q_{XY})$; the two extreme cases follow by a similar argument.
	Since $Q_{XY}^\lambda \in {\cal E}(\theta^\san{xy}(Q))$, it has the same values as $Q_{XY}$ at the $\theta^{\san{xy}}$-coordinate.
	Thus, it holds that (see \eqref{eq:correlation-coordinate}),
	for $x \neq 0$ and $y\neq 0$,
	\begin{align}
	\log \frac{Q_{XY}^\lambda(x,0) Q_{XY}^\lambda(0,y)}{Q_{XY}^\lambda(x,y) Q_{XY}^\lambda(0,0)} \frac{Q_{XY}(x,y) Q_{XY}(0,0)}{Q_{XY}(x,0)Q_{XY}(0,y)} = 0.
	\end{align}
	By noting this fact, we have
	\begin{align}
	\lefteqn{ \log \frac{Q^\lambda_{XY}(x,0)}{Q_{XY}(x,0)} +  \log \frac{Q^\lambda_{XY}(0,y)}{Q_{XY}(0,y)} } \label{eq:identity-proof-0} \\
	&= \log \frac{Q_{XY}^\lambda(x,0) Q_{XY}^\lambda(0,y)}{Q_{XY}(x,0) Q_{XY}(0,y)} \\
	&= \log \frac{Q_{XY}^\lambda(x,y) Q_{XY}^\lambda(0,0)}{Q_{XY}(x,y)Q_{XY}(0,0)} 
	\frac{Q_{XY}^\lambda(x,0) Q_{XY}^\lambda(0,y)}{Q_{XY}^\lambda(x,y) Q_{XY}^\lambda(0,0)} \frac{Q_{XY}(x,y) Q_{XY}(0,0)}{Q_{XY}(x,0)Q_{XY}(0,y)} \\
	&= \log \frac{Q_{XY}^\lambda(x,y) Q_{XY}^\lambda(0,0)}{Q_{XY}(x,y)Q_{XY}(0,0)} \\
	&= \log \frac{Q^\lambda_{XY}(x,y)}{Q_{XY}(x,y)} + \log \frac{Q^\lambda_{XY}(0,0)}{Q_{XY}(0,0)} \label{eq:identity-proof}
	\end{align}
	for $x \neq 0$ and $y \neq 0$; also, the identity \eqref{eq:identity-proof-0}-\eqref{eq:identity-proof} trivially holds when either $x=0$ or $y=0$.
	Similarly, since $P^\lambda_{XY} \in {\cal E}(\theta^\san{xy}(P))$, we have
	\begin{align}
	\log \frac{P^\lambda_{XY}(x,0)}{P_{XY}(x,0)} + \log \frac{P^\lambda_{XY}(0,y)}{P_{XY}(0,y)} = \log \frac{P^\lambda_{XY}(x,y)}{P_{XY}(x,y)} + \log \frac{P^\lambda_{XY}(0,0)}{P_{XY}(0,0)},
	\end{align}
	which together with the identity \eqref{eq:identity-proof-0}-\eqref{eq:identity-proof} imply 
	\begin{align}
	\Lambda_\lambda(x,0) + \Lambda_\lambda(0,y) = \Lambda_\lambda(x,y) + \Lambda_\lambda(0,0).
	\end{align}
	\qed
	
	%%%%%%%%%%%%%%%%%%%%%%%%%%%%
	\subsection{Proof of Remark \ref{remark:parallel-condition}} \label{appendix:proof-remark:parallel-condition}
	
	By noting that $\theta^{\san{xy}}$ components of the pair $P_{XY}^\lambda$ and $P_{XY}$ and the pair
	$Q_{XY}^\lambda$ and $Q_{XY}$ are the same, respectively, and by noting \eqref{eq:natural-parameterization}, we can rewrite \eqref{eq:condition-alignment} as
	\begin{align}
	\lefteqn{ \sum_{i=1}^{d_\san{x}} \big[ \theta^\san{x}_i(P_{XY}^\lambda) - \theta^\san{x}_i(P_{XY}) \big] \delta_i(x)
	 + \sum_{j=1}^{d_\san{y}} \big[ \theta^\san{y}_j(P_{XY}^\lambda) - \theta^\san{y}_j(P_{XY}) \big] \delta_j(y)
	 - \psi(\theta(P_{XY}^\lambda)) + \psi(\theta(P_{XY})) } \\
	 &= a \sum_{i=1}^{d_\san{x}} \big[ \theta^\san{x}_i(Q_{XY}^\lambda) - \theta^\san{x}_i(Q_{XY}) \big] \delta_i(x)
	  + a \sum_{j=1}^{d_\san{y}} \big[ \theta^\san{y}_j(P_{XY}^\lambda) - \theta^\san{y}_j(P_{XY}) \big] \delta_j(y)
	  - a \psi(\theta(Q_{XY}^\lambda)) + a \psi(\theta(Q_{XY})) + b
	\end{align}
	for every $(x,y) \in {\cal X} \times {\cal Y}$. By substituting $(x,y) = (0,0)$, we have
	\begin{align}
	- \psi(\theta(P_{XY}^\lambda)) + \psi(\theta(P_{XY})) = - a \psi(\theta(Q_{XY}^\lambda)) + a \psi(\theta(Q_{XY})) + b.
	\end{align}
	Thus, we have
	\begin{align}
	\lefteqn{ \sum_{i=1}^{d_\san{x}} \big[ \theta^\san{x}_i(P_{XY}^\lambda) - \theta^\san{x}_i(P_{XY}) \big] \delta_i(x)
	 + \sum_{j=1}^{d_\san{y}} \big[ \theta^\san{y}_j(P_{XY}^\lambda) - \theta^\san{y}_j(P_{XY}) \big] \delta_j(y) } \\
	 &=  a \sum_{i=1}^{d_\san{x}} \big[ \theta^\san{x}_i(Q_{XY}^\lambda) - \theta^\san{x}_i(Q_{XY}) \big] \delta_i(x)
	  + a \sum_{j=1}^{d_\san{y}} \big[ \theta^\san{y}_j(P_{XY}^\lambda) - \theta^\san{y}_j(P_{XY}) \big] \delta_j(y),
	\end{align}
	which implies \eqref{eq:parallel-condition-1} and \eqref{eq:parallel-condition-2} since $\delta_i(x)$ and $\delta_j(y)$
	for $i=1,\ldots,d_\san{x}$ and $j=1,\ldots,d_\san{y}$ are linearly independent as functions on ${\cal X}\times {\cal Y}$. \qed
	
	%%%%%%%%%%%%%%%%%%%%%%%%%%%%
	\section*{Acknowledgement}
	
	The author would like to thank Hiroshi Nagaoka for helpful suggestions on information geometry and encouragement,
	Shigeaki Kuzuoka for fruitful discussion, and
	Po-Ning Chen for comments on an earlier version of the manuscript.
	%The author also appreciate the associate editor and anonymous reviewers for valuable comments, which improved 
	%the presentation of the paper. 
	The author is supported in part by the Japan Society for the Promotion of Science (JSPS) KAKENHI under 
	Grant 16H06091.

%%%% Bib %%%%%%%%%%%%%

\bibliographystyle{../../09-04-17-bibtex/IEEEtranS}
\bibliography{../../09-04-17-bibtex/reference.bib}

\end{document}